\documentclass[12pt,a4paper,oneside]{article}

\usepackage{amsmath,amssymb,mathtools,amsthm,amscd,rotating,array,stmaryrd,accents}
\usepackage[authoryear,round]{natbib}
\usepackage{enumitem,multirow}
\usepackage{psfrag,epsfig,boxedminipage}
\usepackage{pgfplots}
\usepackage{dsfont}
\usepackage{longtable}
\usepackage{booktabs}
\usepackage{appendix}
\pgfplotsset{compat=1.11}
\usepackage{float}
\usepackage{tikz}

\usepackage[pdftex]{pict2e} 
\usepackage[11pt]{moresize} 
\usepackage{anyfontsize} 
\usepackage{hyperref} 

\newcounter{AnnahmenCounterA}
\newcounter{AnnahmenCounterB}

\newtheorem*{theorem*}{Argmax Theorem}
\newtheorem{corollary}{Corollary}[section]
\newtheorem{theorem}{Theorem}

\newtheorem{lemma}{Lemma}[section]

\renewenvironment{proof}{{\bfseries Proof.}}

\theoremstyle{definition}

\newcommand{\dt}{\,\mathrm{d}}

\newcommand{\vth}{\vartheta}
\newcommand{\fth}{\boldsymbol{\theta}}

\newcommand{\circled}[1]{\tikz[baseline=(char.base)]{
		\node[shape=circle,draw,inner sep=1pt](char){#1};}}

\begin{document}

\title{A Kolmogorov-Smirnov-Type Test for Dependently Double-Truncated Data}

\author{Anne-Marie Toparkus \& Rafael Wei{\ss}bach \\[2mm] \textit{\footnotesize{Chair of Statistics and Econometrics,}}   \\[-2mm]
        \textit{\footnotesize{Faculty for Economic and Social Sciences,}} \\[-2mm]
        \textit{\footnotesize{University of Rostock}} \\
        }
\date{ }
\maketitle

\renewcommand{\baselinestretch}{1.5}\normalsize

\begin{abstract}
With double-truncated lifespans, we test the hypothesis of a parametric distribution family for the lifespan. The typical finding from demography is an instationary behaviour of the life expectancy, and a copula models the resulting weak dependence of lifespan and the age at truncation. Our main example is the Farlie-Gumbel-Morgenststern copula.  The test is based on Donsker-class arguments and the functional delta method for empirical processes. The assumptions also allow parametric inference, and proofs slightly simplify due to the compact support of the observations. An algorithm with finitely many operations is given for the computation of the test statistic. Simulations becomes necessary for computing the critical value. With the exponential distribution as an example, and for the application to 55{,}000 German double-truncated enterprise lifespans, the constructed Kolmogorov-Smirnov test rejects clearly an age-homogeneous closure hazard. 
 \\[2mm]
\end{abstract}

\section{Introduction}

Truncation can be part of the sampling design for panel data and especially for event histories. Nonparametric contributions for independent left and also double truncation include \cite{And0,kalblawl1989,efron1999,shen2010,doerr2019} or \cite{franchae2019}. (Semi-)Parametric models are estimated with the likelihood, a conditional likelihood or a profile likelihood in \cite{kalblawl1989,moreira2010a,emura2017,doerr2019,doerre2020,weiswied2021} or \cite{weissbachm2021effect}.

More specifically, let the population be defined as all units of a kind with the birth event in a period of length $G$. In interval sampling, we observe units affected by a death event in a study period of length $s$, which for simplicity we assume to directly follow the birth period \cite[see][Figure 1]{weiswied2021}. In order to exclude censoring as additional design defect, the data include the birthdate for each observed unit and it is usual to code it backwards from the start of the study period and denote this ``age when the study starts'' as $T$ \cite[see][Figure 1]{topweis2025}. The independence of the lifetime $X$ and the birthdate equivalent $T$ contradicts the scientific consensus, at least for human mortality, of an increase in life expectancy \cite[see e.g.][]{oepvau2002}. Dependent single and double truncation have also recently been studied \citep{chiou2019,emura2020,morei2021,renxi2022,topweis2025}.

Our contribution is a nonparametric test of whether the lifetime $X$ may be assumed to have some parametric distribution $f^X_{\theta}$, with parameter $\theta$. A worked-out example will concern the exponential distribution. 
We model the dependency with the one-dimensional parameter $\vartheta$ in the Farlie-Gumbel-Morgenstern copula, so that, conditional on the birthdate, the life expectancy may be trended in both directions, and is free of a trend for $\vartheta=0$. 

Neighbouring works for the special case without dependence include \cite{morei2014}, who construct a test for the distribution of $T$. (For simplicity, we assume in this paper that births stem from a homogeneous Poisson process so that no additional parameter is needed.) As well without dependence, \cite{deunaalv2025} constructs a test for the extended questions of a regression function. For dependent data, but for the special case of left truncation, \cite{emura2012} propose a computationally efficient test.  

Our economic application aims at testing whether, under a presumable trend in business demography, enterprise age influences the risk of an enterprise closing operations (due to insolvency or any other reason). As the population, we consider German enterprises founded in the first quarter century after the German reunification. As data we use 55{,}279 enterprise lifespans, double-truncated as a result of their reported closures in the years 2014 to 2016. The dataset of German enterprise lifespans, previously analyzed parametrically e.g. in \cite{weiswied2021} and \cite{topweis2025}, is re-examined here. Whereas \cite{weiswied2021} focused on estimating the expected lifetime under independent truncation, \cite{topweis2025} expanded the approach to show that the conditional expected lifespan decreases. 

In Section \ref{popmodass}, we define the bivariate distribution model of lifespan and truncation age in the population, formalize the sampling design, derive the observation probability and finally describe the test statistic. Section \ref{seccomsta} reduces the computation of the statistic to finitely many operations. Section \ref{testdist} studies the asymptotic distribution of the test statistic under both truncation independence and dependence, including the derivation of an algorithm for the critical value. 
 Section \ref{secparmod} applies the results to a German dataset.

\section{Population model, sampling design and test statistic} \label{popmodass}
The population of interest is defined as the units originated (``born'') within a time window of length $G$, where $0 < G < \infty$. For each unit are measured, (i) its lifespan $X \in \mathbb{R}^+_0$, with conceivable density $f_{\theta}^X$ and cumulative distribution $F^X_{\theta}$, and (ii) its birth date, counted backwards from the end of the population time window, $T \in [0,G]$, with density $f^T$ and cumulative distribution $F^T$. The domain of outcomes is $S:=\mathbb{R}^+_0\times[0,G]$, which we equip with the Borel sigma-field $\mathcal{A}$. We assume dependence between $X$ and $T$ and describe it with the copula $C_{\vth}$, so that the joint distribution function is, according to Sklar's theorem \cite[][Theorem 2.3.3]{nels2006}, $F_{\fth}(x,t):=C_{\vth}(F^X_{\theta}(x),F^T(t))$, with density $f_{\fth}$ and probability measure $\mathbb{P}_{\fth}$. The distribution depends on the parameter $\fth:=(\theta, \vth)^\textup{T} \in \Theta$, which we simplify here to dimension two. (The arguments in this paper can be extended to the distribution of $T$ so as to have a parameter as well.) We assume the existence of a latent sample, a simple random sample consisting of $n\in \mathbb{N}$ random variables $\Omega \to \mathcal{A}$ with $\mathbb{P}_{\fth_0}$, i.e. $(X_i,T_i)^\text{T}, \, i=1, \ldots, n$, with the unknown true parameter $\fth_0 \in \Theta$.

Let the study start immediately after the birth period and include a sample unit if the unit has not died before the study (left truncation). Let the study continue for say $s>0$ years, and also exclude any unit that dies afterwards (right truncation) \cite[for more motivation see][]{weiswied2021}. Note that $T$ now has the interpretation as age-at-study-start. Formally, $(X_i,T_i)^\text{T}$ is observed if it is in the parallelogram 
\begin{equation*}
	D:=\{(x,t)^\text{T}|0<t\leq x\leq t+s,t\leq G\}\subset S.
\end{equation*}
An observation may be denoted by $(X_j^{\textup{obs}},T_j^{\textup{obs}})^\text{T}$ with ${j=1,\ldots,M_n\leq n}$, but note that these are not measurable mappings $\Omega \to \mathcal{A}$. Furthermore, note also that $M_n=\sum_{i=1}^n \mathds{1}_{\{(X_i,T_i)^\text{T} \in D\}}$ is random. The probability that the $i$th sample unit is observed can be calculated through
\begin{equation*}
	\alpha_{\fth}:=\mathbb{P}_{\fth}(T_i\leq X_i\leq T_i +s)=\int_0^G\int_t^{t+s} f_{\fth}(x,t) \dt x \dt t.
\end{equation*}
For a left-truncated design, \citet[][Theorem 1]{emura2020} provides a method  for reducing the integral to one dimension. This saves computational time, and \citet[][Equation 2]{topweis2025} derives for the double-truncated design 
\begin{equation*}
	\alpha_{\fth}=\int_0^1 \left[c_u\{F^T((F^X_{\theta})^{-1}(u))\}-c_u\{F^T((F^X_{\theta})^{-1}(u))-F^T(s)\} \right]\dt u,
\end{equation*}
where $c_u(v):=\frac{\partial C_{\vth}(u,v)}{\partial u}$. We must assume that
\begin{enumerate}[label=(A\arabic*)]
	\item \label{A1} $0 < \alpha_{\fth} <1$, for all $\fth \in \Theta$.
\end{enumerate}
We need a Fr\'{e}chet derivative later, and its simple calculation will require the possibility of ``differentiation under the integral sign''. We now formulate assumptions about the model that (i) enable such simple calculation, (ii) can easily be verified, and (iii) do not rule out any model which we have encountered so far. Recall that a partially differentiable function with continuous derivatives is totally differentiable. 
Therefore, we write, in short, ``continuously differentiable''. According to \citet[][Lemma 16.2]{bauer2001}, we need the following assumptions. (As usual, $\dot{g}_{\theta}$ and $\dot{g}_{\vth}$ denote the partial derivatives $\frac{\partial}{\partial \theta} g_{\fth}$ and $\frac{\partial}{\partial \vth} g_{\fth}$, respectively, and  $\dot{g}_{\fth}:=(\dot{g}_{\theta},\dot{g}_{\vth})^{\textup{T}}$ is the gradient, for whatever function $g_{\fth}$.)
\begin{enumerate}[label=(A\arabic*)]
	\setcounter{enumi}{1} 
	\item \label{A2} The function $(x,t) \mapsto f_{\fth}(x,t)$ is integrable for each $\fth \in \Theta$; the map \linebreak $\fth \mapsto f_{\fth}(x,t)$ is continuously differentiable on $\Theta$ for each $(x,t) \in S$ with the derivatives $\dot{f}_{\fth}:=(\dot{f}_{\theta},\dot{f}_{\vth})^\textup{T}$; there is an integrable function $h_{\fth}\geq 0$ on $S$ such that $\Vert \dot{f}_{\fth}(x,t)\Vert \leq h_{\fth}$ for all $(x,t) \in S$ and $\fth \in \Theta$.
\end{enumerate}
Let $\hat{\fth}_n$ be a reasonable estimator for $\fth_0$. It is not required that $\hat{\fth}_n$ be a maximum likelihood estimator; however consistency and asymptotic linearity are essential.
\begin{enumerate}[label=(A\arabic*)]
	\setcounter{enumi}{2} 
	\item \label{A3} It is $\hat{\fth}_n \overset{P} \longrightarrow \fth_0$ as $n \to \infty$.
	\item \label{A4} With some $\phi_{\fth_0}=(\phi_{\fth_0,1},\phi_{\fth_0,2})^{\textup{T}}$, the estimating sequence is asymptotically linear:  $\sqrt{n}(\hat{\fth}_n-\fth_0)=\frac{1}{\sqrt{n}} \sum_{i=1}^n \phi_{\fth_0}(X_i,T_i)+o_{\mathbb{P}_{\fth_0}}(1)$
	\item \label{A5} The map $\fth \mapsto \phi_{\fth}(x,t)$ is continuously differentiable and fulfills \newline $\sup_{(x,t) \in S} \Vert \phi_{\fth_0}(x,t)\Vert < \infty$, $\mathbb{E}_{\fth_0}[\phi_{\fth_0}(X_i,T_i)\mathds{1}_{ D}(X_i,T_i)]
	=\mathbf{0}$, \newline $\mathbb{E}_{\fth_0}[\phi_{\fth_0}(X_i,T_i)]=\mathbf{0}$ and $\mathbb{E}_{\fth_0}[\Vert\phi_{\fth_0}(X_i,T_i)\Vert^2]<\infty$.
\end{enumerate}

Approaches to derive the function $\phi_{\fth}$ for independent and dependent double truncation can be found in \cite{weiswied2021} and \cite{topweis2025}, respectively.

Our aim is to assess the goodness of fit of the underlying distributional assumptions.
\begin{equation}\label{hypoth}
	\begin{aligned}
		&H_0: \, \text{The distribution of } (X_i,T_i)^{\text{T}} \text{ derives from the parametric family}\\
		& \, \qquad \{F_{\fth}=C_{\vth}(F_{\theta}^X,F^T) \, |\,  \fth \in \Theta \}.\\ 
		&H_1: \, \text{The distribution of } (X_i,T_i)^{\text{T}} \text{ does not derive from the parametric}\\
		& \, \qquad\text{family }\{F_{\fth}=C_{\vth}(F_{\theta}^X,F^T)\,|\, \fth \in \Theta \}.
	\end{aligned}
\end{equation}

The  Kolmogorov-Smirnov (KS) test statistic, now in two dimensions, compares the empirical with a theoretical distribution function. The theoretical shall be $F_{\fth}$, but the comparison is only possible conditional on observation, i.e. with outcome in $D$. The distribution of the observation $(X_j^{\textup{obs}}, T_j^{\textup{obs}})^\text{T}$ is

\begin{equation} \label{distobs}
	F_{\fth}^{\textup{obs}}(x,t):=\mathbb{P}_{\fth}((X_i,T_i)^\text{T}\in [0,x]\times[0,t] \cap D)/\alpha_{\fth}.
\end{equation}
We adapt the \textit{empirical distribution} as the discrete random measure given by $\mathbb{P}_n(A):=n^{-1}\sum_{i=1}^n \epsilon_{\binom{X_i}{T_i}}(A)$ for $A:=[0,x] \times [0,t]$ to truncation on $D$, that is $\mathbb{P}_{n,D}(A):=n^{-1}\sum_{i=1}^n \epsilon_{\binom{X_i}{T_i}}(A \cap D)$, and compare it with its theoretical analogue 
$\alpha_{\fth_0} F_{\fth_0}^{\textup{obs}}(x,t)$. The Kolmogorov-Smirnov (KS) test statistic for known $n$ and a known true parameter $\fth_0$ is 
\begin{equation}\label{dreieck}
	\sup_{(x,t)^{\textup{T}} \in S} \sqrt{n} \, \big\vert \mathbb{P}_{n,D}([0,x]\times[0,t])- \alpha_{\fth_0}F_{\fth_0}^{\textup{obs}}(x,t) \big\vert .
\end{equation}
We also wish to take into account that neither $\fth_0$ nor $n$ are observed in practical applications. In a parametric analysis, $n$ can be profiled out and does not need to be estimated \cite[see][]{weiswied2021}. However, here estimation becomes mandatory and we use $\hat{n}:=M_n/\alpha_{\hat{\fth}_n}$, the maximum likelihood estimator. Inserting the estimators leads to the composite KS statistic:
\begin{equation}\label{e10}
	\frac{\sqrt{M_n}}{\sqrt{\alpha_{\hat{\fth}_n}}}	\sup_{(x,t)^{\textup{T}} \in S} \left\vert \mathbb{P}_{\hat{n},D}\big([0,x]\times[0,t]\big)- \alpha_{\hat{\fth}_n} F_{\hat{\fth}_n}^{\textup{obs}}(x,t)\right\vert
\end{equation}
with 
$\mathbb{P}_{\hat{n},D}:= 	\frac{\alpha_{\hat{\fth}_n}}{M_n} \sum_{i=1}^n \epsilon_{\binom{X_i}{T_i}}(A \cap D)$.

\section{Computation of the test statistic} \label{seccomsta}

Recall \eqref{distobs}, using the observations, the test statistic \eqref{e10} can then be rewritten  as
\begin{multline*}
\frac{\sqrt{M_n}}{\sqrt{\alpha_{\hat{\boldsymbol{\theta}}_n}}}	\sup_{(x,t)^{\textup{T}} \in S} \left\vert \mathbb{P}_{\hat{n},D}\big([0,x]\times[0,t]\big)-\mathbb{P}_{\hat{\theta}_n}\big((X_i,T_i)^\textup{T} \in [0,x]\times[0,t]\cap	D\big)\right\vert \\
	= \sqrt{\alpha_{\hat{\boldsymbol{\theta}}_n}}\sqrt{M_n} \sup_{(x,t)^{\textup{T}} \in D} \left\vert \frac{1}{M_n} \sum_{j=1}^{M_n} 		\epsilon_{\big(\genfrac{}{}{0pt}{}{X_j^{\textup{obs}}}{T_j^{\textup{obs}}}\big)}([0,x]\times[0,t]) - F^{\textup{obs}}_{\hat{\boldsymbol{\theta}}_n}(x,t) \right\vert .
\end{multline*}

  In the event case of a one-dimensional situation, due to monotonicity of the CDF, the univariate Kolmogorov-Smirnov test statistic is obtained by computing the distance at the discontinuity points of the empirical distribution. This includes the right-side limits at these points \cite[see e.g.][]{holl2014}. The monotonicity arguments also hold true in the bivariate case, with the two dimensions $X_j^{\textup{obs}}$ and $T_j^{\textup{obs}}$. However, the set of discontinuities is not finite, but, again by virtue of monotonicity arguments, the calculation of the KS statistic requires only a finite number of points when computing \eqref{e10}. \cite{justel1997} present details for two independent uniform distributions, with the unit square as support. We adjust their arguments to the parallelogram $D$ as support.   

Let $(x_1^{\textup{obs}},t_1^{\textup{obs}})^{\textup{T}}, \ldots, (x_{m_n}^{\textup{obs}},t_{m_n}^{\textup{obs}})^{\textup{T}}$ be the realization of a truncated sample and let $F^{\textup{obs}}_{m_n}$ denote the realized empirical distribution function. Define left-side and right-side differences 
\begin{eqnarray*}
D_{m_n}^+(x,t) & :=& F^{\textup{obs}}_{m_n} (x, t)-F^{\textup{obs}}_{\hat{\boldsymbol{\theta}}_n}(x,t) \quad \text{and} \\
 D_{m_n}^-(x,t) & :=& F^{\textup{obs}}_{\hat{\boldsymbol{\theta}}_n}(x,t)- F^{\textup{obs}}_{m_n}(x, t).
 \end{eqnarray*}

 Then, the KS statistic \eqref{e10} needs the maximum of $\delta_{m_n}^+:= \sup_{(x,t)^{\textup{T}} \in D} D_{m_n}^+(x,t)$ and $\delta_{m_n}^-:= \sup_{(x,t)^{\textup{T}} \in D} D_{m_n}^-(x, t)$. Similar to \cite{justel1997}, we can obtain the left-side difference $\delta_{m_n}^{+}$ as the maximum of $D_{m_n}^+(x,t)$ evaluated at $(0,0)$, the observation points (see $\bullet$ in Figure \ref{abb1}) and the intersections originating from discordant points,  
 $$I:=\{(x_j^{\textup{obs}}, t_{j'}^{\textup{obs}}) \vert \,x_{j'}^{\textup{obs}} < x_j^{\textup{obs}}, \; t_{j'}^{\textup{obs}} > t_j^{\textup{obs}}; \, j, j'=1, \ldots, m_n\}$$ 
 (see $\times$ in Figure \ref{abb1}). For $\delta_{m_n}^-$ we need to evaluate the difference  $F^{\textup{obs}}_{\hat{\boldsymbol{\theta}}_n}(x,t)- F^{\textup{obs}}_{m_n}(x^-, t^-)$ with the definition  $F^{\textup{obs}}_{m_n}(x^-, t^-):= \lim_{\delta \to 0}F^{X^{\textup{obs}}, T^{\textup{obs}}}_{m_n}(x- \delta, t-\delta)$, at $(G+s,G)$, the intersection points and the projection points of the observed points onto the upper and right edge of $D$, 
 $$P:=\{(x_j^{\textup{obs}}, \min(x_j^{\textup{obs}}, G)) \vert j=1, \ldots, m_n\} \cup \{(t_j^{\textup{obs}}+s, t^{\textup{obs}}_j) \vert j=1, \ldots, m_n\}$$ 
 (see $\square$ in Figure \ref{abb1}), and the three corners $(G,0), (s,G), (G+s,G)$ (see $\triangle$ in Figure \ref{abb1}).

	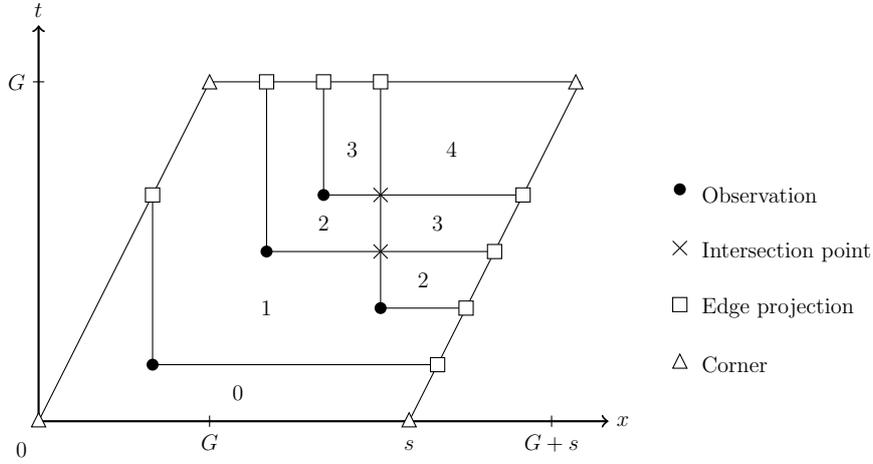
\begin{figure}[H]
	\centering
	\begin{tikzpicture}[scale=0.75, every node/.style={scale=0.7}] 
		\draw[->, thick] (0,0) -- (0,7) node[above] {$t$};
		\draw[->, thick] (0,0) -- (10,0) node[right] {$x$};
		\node at (-0.3,-0.5) {0};
		\draw (-0.1,6) -- (0.1,6) node[left] at (-0.1,6) {$G$};
		\draw (3,-0.1) -- (3,0.1) node[below] at (3,-0.1) {$G$};
		\draw (6.5,-0.1) -- (6.5,0.1) node[below] at (6.5,-0.2) {$s$};
		\draw (9,-0.1) -- (9,0.1) node[below] at (9,-0.1) {$G+s$};

		\foreach \x/\y in {2/1,4/3,6/2,5/4} {
			\fill (\x,\y) circle (3pt);
		}
		
		\foreach \x/\y in {6/3,6/4} {
			\draw (\x-0.125,\y-0.125) -- (\x+0.125,\y+0.125);
			\draw (\x-0.125,\y+0.125) -- (\x+0.125,\y-0.125);
		}
		
		\draw (2,1) -- (2,3.875);
		\draw (4,3) -- (4,5.875);
		\draw (5,4) -- (5,5.875);
		\draw (6,2) -- (6,5.875);
		
		\draw (2,1) -- (6.875,1);
		\draw (6,2) -- (7.375,2);
		\draw (4,3) -- (7.875,3);
		\draw (5,4) -- (8.375,4);
		\draw (3,6) -- (9.5,6);
		
		\draw (0,0) -- (3,6);
		
		\draw (6.5,0) -- (9.5,6);
		
		\foreach \x/\y in {1.875/3.875, 3.875/5.875, 4.875/5.875, 5.875/5.875,
			8.375/3.875, 7.875/2.875, 7.375/1.875, 6.875/0.875} {
			\draw[fill=white] (\x,\y) rectangle ++(0.25,0.25);
		}
		
		\foreach \x/\y in {9.3/5.875,6.375/-0.1,2.875/5.875,-0.125/-0.1} {
			\draw[fill=white] (\x,\y) -- ++(0.25,0)
			-- ++(-0.125,0.25)
			-- ++(-0.125,-0.25)
			-- cycle;
		}
		
		\node at (3.5,0.5) {0};
		\node at (4,2) {1};
		\node at (6.75,2.5) {2};
		\node at (5,3.5) {2};
		\node at (7,3.5) {3};
		\node at (5.5,4.8) {3};
		\node at (7.25,4.8) {4};
		
		\fill (11.25,4.1) circle (3pt);
		\node[anchor=west] at (11.5,4) {Observation};
		
		\draw (11.125,2.937) -- (11.375,3.187);
		\draw (11.125,3.187) -- (11.375,2.937);
		\node[anchor=west] at (11.5,3) {Intersection point};
		
		\draw (11.125,1.937) rectangle ++(0.25,0.25);
		\node[anchor=west] at (11.5,2) {Edge projection};
		
		\draw (11.125,0.937) -- ++(0.25,0) -- ++(-0.125,0.25) -- ++(-0.125,-0.25) -- cycle;
		\node[anchor=west] at (11.5,1) {Corner};
		
	\end{tikzpicture}
	\caption{Two-dimensional empirical distribution function $F^{\textup{obs}}_{m_n}$ for $m_n=4$ (multiplied by 4) with support $D$}
\label{abb1}
\end{figure}

Depending on the sample, $3m_n$ to $3m_n+\binom{m_n}{2}$ evaluations are required. Fortunately, the evaluations for $D_{m_n}^-$ can be traced back to  $D_{m_n}^+$. For the set of intersection points $I$, it is for instance: 
\begin{multline*}
	\max_{(x,t)^{\textup{T}} \in I}(F^{X^{\textup{obs}}, T^{\textup{obs}}}_{\hat{\boldsymbol{\theta}}_n}(x,t)-F^{X^{\textup{obs}},  T^{\textup{obs}}}_{m_n}(x^-, t^-)) \\ =-\min_{(x,t)^{\textup{T}} \in I}(F^{X^{\textup{obs}}, T^{\textup{obs}}}_{m_n}(x^-, t^-)-F^{X^{\textup{obs}}, T^{\textup{obs}}}_{\hat{\boldsymbol{\theta}}_n}(x, t))\\
	=-\min_{(x,t)^{\textup{T}} \in I}\left[\frac{1}{m_n}\sum_{j=1}^{m_n}\epsilon_{\binom{x_j^{\textup{obs}}}{t_j^{\textup{obs}}}}\big([0,x)\times[0,t) \big) -F^{X^{\textup{obs}}, T^{\textup{obs}}}_{\hat{\boldsymbol{\theta}}_n}(x, t)\right]\\
	=\frac{2}{m_n}-\min_{(x,t)^{\textup{T}} \in I}\left[\frac{1}{m_n}\sum_{j=1}^{m_n}\epsilon_{\binom{x_j^{\textup{obs}}}{ t_j^{\textup{obs}}}}\big([0,x]\times[0,t] \big) -F^{X^{\textup{obs}}, T^{\textup{obs}}}_{\hat{\boldsymbol{\theta}}_n}(x, t)\right]\\
	=\frac{2}{m_n}-\min_{(x,t)^{\textup{T}} \in I}\left[F^{X^{\textup{obs}}, T^{\textup{obs}}}_{m_n}(x, t)-F^{X^{\textup{obs}}, T^{\textup{obs}}}_{\hat{\boldsymbol{\theta}}_n}(x,t)\right].
\end{multline*}

Therefore, the test statistic can be obtained from the following calculations:

\rule{\linewidth}{0.5pt}
\hypertarget{h4}{\textbf{Algorithm 1:}} Computation of KS test statistic \eqref{e10} \vspace{-3mm}\\
\rule{\linewidth}{0.5pt}
\begin{enumerate}
	\setlength{\itemsep}{-6pt}
	\item Compute $\delta_1:=\max_{j=1, \ldots, m_n} D_{m_n}^+(x_j^{\textup{obs}}, t_j^{\textup{obs}})$ (``left-maximum distance in the observations'') 
	\item Determine set $I$ (``intersection points''). 
	\item Compute $\delta_2:=\max_{(x,t)^{\textup{T}} \in I}\{D_{m_n}^+(x,t)\}$. (``left-maximum distance in intersection points'') 
	\item Compute $\delta_3:=2/m_n-\min_{(x,t)^{\textup{T}} \in I}\{D_{m_n}^+(x, t)\}$. (``right-maximum distance in intersection points'') 
	\item  Compute $\delta_4:=1/m_n-\min_{(x,t)^{\textup{T}} \in P} D_{m_n}^+(x, t)$. (``right-maximum distance in projections of observations onto  right and upper edge'')
	\item Compute $\delta_5:=\max(D_{m_n}^+(0,0), - D_{m_n}^+(G,0), - D_{m_n}^+(s,G), - D_{m_n}^+(G+s,G)$. (``left- or right-maximum distance in corners'')
	\item Compute $\max\{\delta_1, \delta_2, \delta_3, \delta_4, \delta_5\}$. 
	\item Multiply by $\sqrt{\alpha_{\hat{\boldsymbol{\theta}}_n}}\sqrt{m_n}$.\vspace{-5mm}
\end{enumerate}
\rule{\linewidth}{0.5pt}\\


\section{Distribution of the test statistic under $\mathbf{H_0}$} \label{testdist}

After the computation of our test statistic, we now derive the critical value, i.e. the quantile of \eqref{e10} under $H_0$ of \eqref{hypoth}. Section \ref{3eins} treats latent sample size and parameter as known. Section \ref{3zwei} considers $n$ to be known but $\boldsymbol{\theta}_0$ to be estimated and Section \ref{3drei} finally treats both simultaneously as unknown.  

\subsection{Limit theorems for truncated data} \label{3eins}

Consider set $B \subset S$ as a ``point''; then $\mathbb{P}_{n,D}(B)$ and $\sqrt{n} \mathbb{P}_{n,D}(B)$ converge pointwise by the LLN and the CLT. An important set will be $[0,x]\times[0,t]$. The test statistic \eqref{e10} includes a supremum, so that convergence needs to be uniform, or in other words, functional. (Note already that a corresponding supremum will finally not be one of points in a set, neither one of sets within some algebra, but a supremum of functions $\mathds{1}_{[0, x]\times [0,t] \cap D}$.) We switch to the functional-valued notation (see \citet[][Sect. 19.2]{vaart1998} or \citet[][Sect. 2.1]{vawell1996}).

 The empirical distribution $\mathbb{P}_n$ has been defined in Section \ref{popmodass}, and for a given class $\mathcal{G}$ of Borel-measurable functions $g: S \to \mathbb{R}$, denote $\mathbb{P}_n (g)$ the expectation of $g$ under $\mathbb{P}_n$, i.e.
\begin{equation*}
	g \mapsto \mathbb{P}_n (g):=\frac{1}{n} \sum_{i=1}^n g(X_i,T_i).
\end{equation*}
Furthermore, define $\mathbb{E}_{\fth}(g):=\mathbb{E}_{\fth}(g(X_i,T_i))=\int_{\Omega} g(X_i,T_i) \dt \mathbb{P}_{\fth}$. The centered and scaled version of the mapping $g \mapsto \mathbb{P}_n (g)$ is the (functional-valued) \textit{empirical process} evaluated in $g$, that is:
\begin{equation*}
	g \mapsto \mathbb{G}_n( g):=\sqrt{n}(\mathbb{P}_n-\mathbb{E}_{\fth_0})(g)=\frac{1}{\sqrt{n}} \sum_{i=1}^n (g(X_i,T_i)-\mathbb{E}_{\fth_0}(g)).
\end{equation*}
(We will also use $g \mapsto \mathbb{G}_n( g)$ for two-dimensional functions $g=(g_1,g_2)^\textup{T}$, for which $\mathbb{P}_n(g)=(\mathbb{P}_n(g_1),\mathbb{P}_n(g_2))^\textup{T}$ will be calculated component-wise, as is common for expectations.) Now, recall the following two definitions for a functional-valued definition of the Brownian bridge  \cite[see e.g.][Sect. 2.1]{vawell1996}.	A class $\mathcal{G}$ of measurable functions $g: S \to \mathbb{R}$ is called \textit{Glivenko-Cantelli} if $n^{-1/2}\Vert \mathbb{G}_n( g)\Vert_{\mathcal{G}}:=n^{-1/2} \sup_{g\in \mathcal{G}} \vert \mathbb{G}_n( g) \vert  \longrightarrow 0$,	where the convergence is outer almost surely as $ n \to \infty$.
Further details about outer integrals can be found, for example, in \citet[][Sect. 1.2]{vawell1996}.

	A class $\mathcal{G}$ of measurable functions $g: S\to \mathbb{R}$ is called \textit{Donsker}, if the sequence of processes $\{\mathbb{G}_n(g): \, g \in \mathcal{G}\}$ converges in distribution to a tight limit process $\mathbb{G}$ in the space $\ell^{\infty}(\mathcal{G})$. Then, the limit process is Gaussian with zero mean and covariance function
	\begin{equation}\label{e14}
		\textup{Cov}[\mathbb{G}(g), \mathbb{G}(h)]=\mathbb{E}_{\fth_0}(gh)-\mathbb{E}_{\fth_0}(g)\mathbb{E}_{\fth_0}(h)
	\end{equation}
	and therefore especially a $\mathbb{P}_{\fth_0}$-Brownian bridge $\mathbb{B}_{\mathbb{P}_{\fth_0}}(g):=\mathbb{B}(\mathbb{E}_{\fth_0}(g))$, with $\mathbb{B}$ as real-valued Brownian bridge on $[0,1]$.

According to \citet[][Ex. 19.6]{vaart1998} or \citet[][Ex. 2.5.4]{vawell1996} the class of functions $\mathcal{G}:=\{g_{x,t}=\mathds{1}_{[0, x]\times [0,t] }, \, (x,t) \in S\}$ is Glivenko-Cantelli and Donsker. In view of our truncated process, we need to multiply each function with $\mathds{1}_D$. A certain kind of Lipschitz condition ensures that the desired properties will extend.
\begin{corollary}\label{cor2}\textbf{}\\
	The class of functions $\mathcal{G}_D=\{g_{x,t,D}=\mathds{1}_{[0, x]\times [0,t] }\mathds{1}_D, \, (x,t) \in S\}$ is Glivenko-Cantelli and Donsker
\end{corollary}
\begin{proof}
	As stated above, the class $\mathcal{G}$ is Donsker with $\Vert \mathbb{E}_{\fth_0}\Vert_{\mathcal{G}}\leq 1<\infty$. Furthermore, the function $(x,t) \mapsto \mathds{1}_D(x,t)$ is measurable and uniformly bounded. Although the product of functions from $\mathcal{G}$ and $ \mathds{1}_D$ is not a Lipschitz function, we have
	\begin{align*}
		\vert g_{x_1,t_1,D}(x,t) \mathds{1}_D(x,t) &- g_{x_2,t_2,D}(x,t) \mathds{1}_D(x,t)\vert\\
		&\leq \Vert \mathds{1}_D \Vert_{\infty}\, \vert g_{x_1,t_1,D}(x,t) - g_{x_2,t_2,D}(x,t) \vert
	\end{align*}
	for all $(x,t), (x_1,t_1), (x_2,t_2) \in S$. Hence, \citet[][Condition (2.10.5)]{vawell1996} is satisfied and Theorem 2.10.6 applies therein. Hence, the class $\mathcal{G}_D$ is Donsker. It is also Glivenko-Cantelli, given that every Donsker class is a Glivenko-Cantelli class \cite[][Section 2.1]{vawell1996}. \qed
\end{proof}
In terms of the truncated empirical distribution, this means
\begin{equation*}
	\sup_{(x,t) \in S} \big\vert  \mathbb{P}_{n,D}\big([0,x]\times[0,t]\big)-\mathbb{P}_{\boldsymbol{\theta}_0}\big((X_i,T_i)^\textup{T} \in [0,x]\times[0,t]\cap	D\big) \big\vert\overset{a.s.} \longrightarrow 0
\end{equation*}
and the sequence
\begin{equation*}
	\big\{\sqrt{n}\big(\mathbb{P}_{n,D}\big([0,x]\times[0,t]\big)-\mathbb{P}_{\boldsymbol{\theta}_0}\big((X_i,T_i)^\textup{T} \in [0,x]\times[0,t]\cap	D\big)\big): (x,t) \in S\big\}
\end{equation*}
converges in distribution to a $\mathbb{P}_{\theta_0}$-Brownian bridge $\mathbb{B}_{\mathbb{P}_{\theta_0}}(g_{x,t,D})$ (in $\ell^{\infty}(\mathcal{G_D})$). 

The distribution of the test statistics as supremum can then be simulated by simulating Brownian bridges with the covariance function \eqref{e14} as high-dimensional multivariate Gaussian vectors. The necessary discretisation of the function space $\mathcal{G_D}$ of indicator functions is simple to realise, by choosing a  grid of the support $D$. Then expressions $\mathbb{E}_{\fth_0}(g_{x_i,t_i,D})$ ($i=1,2$ and not to be confused with $i$ as index of the unit in Section \ref{popmodass}) and $\mathbb{E}_{\fth_0}(g_{x_1,t_1,D}g_{x_2,t_2,D})$ must be calculated. 
These will not be given for the preliminary situation here, but the strategy will not change when we account for the unknown $\fth_0$ and $n$, and Appendix \ref{App4} will list the arising expectations. An algorithm summing up the approximating steps will be provided in Section \ref{sec4_4}.


\subsection{Estimation of $\fth_0$ where $\mathbf{n}$ is known (and other preliminaries)}  \label{3zwei}

Write the test statistic \eqref{e10} now in a functional-valued notation. We set 
\begin{equation}\label{e1}
	\mathbb{T}(g_{x,t,D}):=\frac{\sqrt{\alpha_{\hat{\fth}_n}}}{\sqrt{M_n}}n\mathbb{P}_n(g_{x,t,D})-\frac{\sqrt{M_n}}{\sqrt{\alpha_{\hat{\fth}_n}}}\mathbb{E}_{\hat{\fth}_n}(g_{x,t,D}),
\end{equation}
which is to be maximized in absolute terms, i.e. \eqref{e10} is equal to
\begin{equation}\label{e11}
	\sup_{g_{x,t,D} \in \mathcal{G}_D} \left\vert 	\mathbb{T}(g_{x,t,D}) \right\vert .
\end{equation}
Note that $n\mathbb{P}_n(g_{x,t,D})$ does not depend on $n$. In a first step, we determine the asymptotic behaviour of $\mathbb{T}(g_{x,t,D})$ by decomposing it into several sub-problems, that are analyzed separately.

When we add to equation \eqref{e1} the terms
\begin{align*}
	\pm \sqrt{n}\left( \mathbb{P}_n(g_{x,t,D})- \mathbb{E}_{\hat{\fth}_n}(g_{x,t,D})\right) \pm \sqrt{\alpha_{\fth_0}} \left(\frac{n}{\sqrt{M_n}} \mathbb{P}_n(g_{x,t,D})-\frac{\sqrt{n}}{\sqrt{\alpha_{\hat{\fth}_n}}}\mathbb{E}_{\hat{\fth}_n}(g_{x,t,D})\right),
\end{align*}
the following representation results (after factoring out $\mathbb{P}_n(g_{x,t,D})$ and $\mathbb{E}_{\hat{\fth}_n}(g_{x,t,D})$):
\begin{align*}
	&\sqrt{n}\left( \mathbb{P}_n(g_{x,t,D})- \mathbb{E}_{\hat{\fth}_n}(g_{x,t,D})\right)\\
	&+\frac{n}{\sqrt{M_n}} \mathbb{P}_n(g_{x,t,D}) \left[\left(\sqrt{\alpha_{\hat{\fth}_n}}-\sqrt{\alpha_{\fth_0}}\right)- \left(\frac{\sqrt{M_n}}{\sqrt{n}}-\sqrt{\alpha_{\fth_0}}\right)\right]\\
	&+ \frac{\mathbb{E}_{\hat{\fth}_n}(g_{x,t,D})}{\sqrt{\alpha_{\hat{\fth}_n}}} \left[\left(\sqrt{n} \sqrt{\alpha_{\hat{\fth}_n}}-\sqrt{n} \sqrt{\alpha_{\fth_0}}\right)- \left(\sqrt{M_n}-\sqrt{n} \sqrt{\alpha_{\fth_0}}\right)\right]
\end{align*}
Another factoring out in the last two summands of the preceding display leads to
\begin{align}\label{e2}
	\begin{split}
		&\underbrace{\sqrt{n}\left( \mathbb{P}_n(g_{x,t,D})- \mathbb{E}_{\hat{\fth}_n}(g_{x,t,D})\right)}_{\text{\circled{1}}}\\
		&+ \left[\underbrace{\sqrt{n}\left(\sqrt{\alpha_{\hat{\fth}_n}}-\sqrt{\alpha_{\fth_0}}\right)}_{\text{\circled{2}}}-\underbrace {\left(\sqrt{M_n}-{\sqrt{n}}\sqrt{\alpha_{\fth_0}}\right)}_{\text{\circled{3}}}\right]\\
		&\cdot \underbrace{\left[\frac{\sqrt{n}}{\sqrt{M_n}} \mathbb{P}_n(g_{x,t,D})+ \frac{\mathbb{E}_{\hat{\fth}_n}(g_{x,t,D})}{\sqrt{\alpha_{\hat{\fth}_n}}}\right]}_{\text{\circled{4}}}.
	\end{split}
\end{align}
The term \circled{1} in equation \eqref{e2} represents the empirical process in which the parameter $\fth_0$ is estimated, and $n$ is considered as known \cite[a situation also considered by][]{shen2014}.
 Its convergence in distribution follows from the Donsker property in Corollary \ref{cor2} and Assumptions \ref{A3}-\ref{A5}. Specifically, a certain smoothness is required for $\fth \mapsto \mathbb{E}_{\fth}$ as a map from $\Theta$ to $\ell^{\infty}(\mathcal{G}_D)$, which follows from Assumption \ref{A2} (see Appendix \ref{App11}).
\begin{corollary}\label{cor1}
	Under \textup{\ref{A2}} and $H_0$, the map $\fth\in  \Theta \mapsto \mathbb{E}_{\fth} \in \ell^{\infty}(\mathcal{G}_D)$ is Fréchet differentiable at $\fth_0$ with the derivative
	\begin{equation*}
		\dot{\mathbb{E}}_{\fth_0}(g_{x,t,D}):=\int_{[0,x]\times[0,t]\cap D} \dot{f}_{\fth_0}(\tilde{x},\tilde{t}) \dt (\tilde{x}, \tilde{t}).
	\end{equation*}
\end{corollary}
Note that $\dot{g}_{\fth_0}$ is short for $\dot{g}_{\fth}|_{\fth = \fth_0}$.
 Estimating the parameter causes the limit process of term \circled{1} in \eqref{e2} to no longer be a Brownian bridge process, but it remains Gaussian \citep[Theorem 19.23]{vaart1998}. The limit process is derived by decomposing \circled{1} once more, and then using the asymptotic linearity \ref{A4} of the estimating sequence. (We will show the asymptotic linearity for our specific models in Section \ref{secparmod}.) 
The distribution is summarized in Table \ref{tabgenesis} as the second line and compared to the situation of both $\boldsymbol{\theta}_0$ and $n$ being known. 
	\begin{table}[H]
		\caption{Asymptotic distributions of the test statistic before taking supremum, \eqref{dreieck}, for increasing unavailability of $\boldsymbol{\theta}_0$ and $n$}
		\label{tabgenesis}
		\begin{center}
			\begin{tabular}{cc}
				\toprule
				Test statistic \eqref{e1} & Gaussian process representation  	 \\ \midrule[1.3pt]
				$\boldsymbol{\theta}_0,n$ & $\mathbb{B}_{\mathbb{P}_{\boldsymbol{\theta}_0}}$\\ 
				$\hat{\boldsymbol{\theta}}_n, n$ & $\mathbb{B}_{\mathbb{P}_{\boldsymbol{\theta}_0}}-\mathbb{B}_{\mathbb{P}_{\boldsymbol{\theta}_0}}( \phi_{\boldsymbol{\theta}_0}) \cdot  \mathbb{\dot{E}}_{\boldsymbol{\theta}_0}$ \\ 
				$\boldsymbol{\theta}_0, \hat{n}$  & $\mathbb{G}^{\mathds{1}}_{\mathbb{P}_{\theta_0}}:= \mathbb{B}_{\mathbb{P}_{\boldsymbol{\theta}_0}} -\mathbb{B}_{\mathbb{P}_{\boldsymbol{\theta}_0}}(\mathds{1}_{D}) \cdot \mathbb{E}_{\boldsymbol{\theta}_0}/\alpha_{\boldsymbol{\theta}_0}$  \\
				$\hat{\boldsymbol{\theta}}_n, \hat{n}$ & $\mathbb{H}_{\mathbb{P}_{\boldsymbol{\theta}_0}}=\mathbb{G}^{\mathds{1}}_{\mathbb{P}_{\boldsymbol{\theta}_0}}+\mathbb{B}_{\mathbb{P}_{\boldsymbol{\theta}_0}}(\phi_{\boldsymbol{\theta}_0}) \cdot \left[\frac{\dot{\alpha}_{\boldsymbol{\theta}_0}}{\alpha_{\boldsymbol{\theta}_0}}\mathbb{E}_{\boldsymbol{\theta}_0}-\dot{\mathbb{E}}_{\boldsymbol{\theta}_0}\right]$ \\
				\bottomrule[1.3pt]	
			\end{tabular}
		\end{center}
	\end{table}

Let us prepare for the estimation of $n$. With the mentioned asymptotic linearity, we can also rewrite term \circled{2} of equation \eqref{e2}. (The proof can be found in Appendix \ref{App12}.)

\begin{lemma}\label{l2}\textbf{}\\
	With Assumptions \textup{\ref{A1}-\ref{A5}} and under $H_0$, the asymptotic linearity
	\begin{equation}\label{e3}
		 \sqrt{n}( \sqrt{\alpha_{\hat{\fth}_n}}- \sqrt{\alpha_{\fth_0}})=\frac{1}{\sqrt{n}} \sum_{i=1}^n \phi_{\fth_0}(X_i,T_i)^{\textup{T}}\cdot \frac{\dot{\alpha}_{\fth_0}}{2\sqrt{\alpha_{\fth_0}}}+o_{\mathbb{P}_{\fth_0}}(1)
	\end{equation}
	holds.
\end{lemma}

The convergence in distribution of the univariate sequence in term \circled{3} of equation \eqref{e2} can be obtained with the delta method \citep[see e.g.][Section 3.1]{vaart1998}. Note that, because $M_n$ is a binomial random variable,  $\sqrt{n}\left(\frac{1}{n}M_n-\alpha_{\fth_0}\right)\to \mathcal{N}\left(0, \alpha_{\fth_0}(1-\alpha_{\fth_0})\right)$ as $n \to \infty$. Thus, we select $x \mapsto\sqrt{x}$ as a function for the delta method, which is differentiable with $(2\sqrt{x})^{-1}$, $x>0$, so that
\begin{equation}\label{e4}
		\sqrt{n}\left(\sqrt{M_n/n}-\sqrt{\alpha_{\fth_0}}\right)
		\to \mathcal{N}(0,(1-\alpha_{\fth_0})/4).
\end{equation}
The following lemma provides an approximation of term \circled{4}, the last remaining in equation \eqref{e2}. The proof is provided in Appendix \ref{App13}.
\begin{lemma}\label{l1}\textbf{}\\
	Under Assumptions \textup{\ref{A1}-\ref{A3}} and $H_0$, it is
	\begin{equation*}
		\left\Vert \frac{\sqrt{n}}{\sqrt{M_n}} \mathbb{P}_n(g_{x,t,D})+ \frac{\mathbb{E}_{\hat{\fth}_n}(g_{x,t,D})}{\sqrt{\alpha_{\hat{\fth}_n}}} - 2 \cdot \frac{\mathbb{E}_{\fth_0}(g_{x,t,D})}{\sqrt{\alpha_{\fth_0}}} \right\Vert_{\mathcal{G}_D}\to 0
	\end{equation*}
	in probability $\mathbb{P}_{\fth_0}$ as $n \to \infty$. 
\end{lemma}


\subsection{Estimation of $\fth_0$ and $\mathbf{n}$}  \label{3drei}

We now combine the four series in decomposition \eqref{e2} of \eqref{e1}.
\begin{theorem}\label{th1}\textbf{}\\
	Suppose that \textup{\ref{A1}} - \textup{\ref{A5}} hold, then the sequence
	\begin{equation*}
		\left\{\frac{\sqrt{\alpha_{\hat{\fth}_n}}}{\sqrt{M_n}}n\mathbb{P}_n(g_{x,t,D})-\frac{\sqrt{M_n}}{\sqrt{\alpha_{\hat{\fth}_n}}}\mathbb{E}_{\hat{\fth}_n}(g_{x,t,D}): g_{x,t,D} \in \mathcal{G}_D\right\}
	\end{equation*}
	converges in distribution in $\ell^{\infty}(\mathcal{G}_D)$ under $H_0$ to the Gaussian process  with values
	\begin{align*}
		\mathbb{H}_{\mathbb{P}_{\fth_0}}(g_{x,t,D}):=\,& \mathbb{B}_{\mathbb{P}_{\fth_0}} (g_{x,t,D})-\left(\mathbb{B}_{\mathbb{P}_{\fth_0}}(\phi_{\fth_0,1}),\mathbb{B}_{\mathbb{P}_{\fth_0}}(\phi_{\fth_0,2})\right) \cdot \dot{\mathbb{E}}_{\fth_0}(g_{x,t,D})\\
		&+\left[\left(\mathbb{B}_{\mathbb{P}_{\fth_0}}( \phi_{\fth_0,1}), \mathbb{B}_{\mathbb{P}_{\fth_0}}( \phi_{\fth_0,2})\right)  \cdot \dot{\alpha}_{\fth_0}
		- \mathbb{B}_{\mathbb{P}_{\fth_0}}(\mathds{1}_D) \right]\cdot \frac{\mathbb{E}_{\fth_0}(g_{x,t,D})}{\alpha_{\fth_0}}
	\end{align*}
	with $g_{x,t,D} \in \mathcal{G}_D$.
\end{theorem}
Before the proof, we would like to stress that the practically necessary covariance function will be derived in detail later. 

\begin{proof}
 The first step involves a series of approximations in terms of the
 	norm $\ell^{\infty}(\mathcal{G_D})$. For this purpose, the Lemma \ref{l1} is applied to term \circled{4} in \eqref{e2} and Lemma \ref{l2} is applied to term \circled{2} of \eqref{e2}:
	\begin{align*}
		\eqref{e1}=\eqref{e2}=&\,\sqrt{n}\left( \mathbb{P}_n(g_{x,t,D})- \mathbb{E}_{\hat{\fth}_n}(g_{x,t,D})\right)\\
		&+\left[ \frac{1}{\sqrt{n}} \sum_{i=1}^n \phi_{\fth_0}(X_i,T_i)^{\textup{T}}\cdot \frac{\dot{\alpha}_{\fth_0}}{2\sqrt{\alpha_{\fth_0}}}-\left(\sqrt{M_n}-\sqrt{n} \sqrt{\alpha_{\fth_0}}\right)\right]\\
		&\cdot \left[2\, \mathbb{E}_{\fth_0}(g_{x,t,D})/\sqrt{\alpha_{\fth_0}}\right]+o_{\mathbb{P}_{\fth_0}}(1)
	\end{align*}
	The linear representation of Lemma \ref{l2}, as well as the observed sample size $M_n$ and the observation probability $\alpha_{\fth_0}$, can be expressed in terms of the empirical process too, i.e.
	\begin{align*}
		&\frac{1}{\sqrt{n}} \sum_{i=1}^n \phi_{\fth_0}(X_i,T_i)=\sqrt{n}\big(\mathbb{P}_n(\phi_{\fth_0})-\mathbb{E}_{\fth_0}(\phi_{\fth_0})\big) \quad \text{and}\\
		&\sqrt{M_n}-\sqrt{n} \sqrt{\alpha_{\fth_0}}=\sqrt{n}\left(\sqrt{\mathbb{P}_n(\mathds{1}_D)}
		-\sqrt{\mathbb{E}_{\fth_0}(\mathds{1}_D)}\right),
	\end{align*}
	where $\mathbb{E}_{\fth_0}(\phi_{\fth_0})=\boldsymbol{0}$ (as per \ref{A5}). Inserting these two terms results in
	\begin{align}\label{e8}
		\begin{split}
		\eqref{e1}=&\,\sqrt{n}\left( \mathbb{P}_n(g_{x,t,D})- \mathbb{E}_{\hat{\fth}_n}(g_{x,t,D})\right)\\
		&+\left[\sqrt{n}\big(\mathbb{P}_n(\phi_{\fth_0})-\mathbb{E}_{\fth_0}(\phi_{\fth_0})\big)^{\textup{T}}\cdot \frac{\dot{\alpha}_{\fth_0}}{2\sqrt{\alpha_{\fth_0}}}-\sqrt{n}\left(\sqrt{\mathbb{P}_n(\mathds{1}_D)}
		-\sqrt{\mathbb{E}_{\fth_0}(\mathds{1}_D)}\right)\right]\\
		&\cdot \left[2\, \mathbb{E}_{\fth_0}(g_{x,t,D})/\sqrt{\alpha_{\fth_0}}\right]+o_{\mathbb{P}_{\fth_0}}(1).
		\end{split}
	\end{align} 
	The third and final approximation is assigned to term \circled{1} of the expression \eqref{e2}. As previously mentioned, the convergence in distribution of this expression was proven in \citet[][Theorem 19.23]{vaart1998}. The proof of this theorem is based on the following decomposition: \begin{align*}
		\sqrt{n}(&\mathbb{P}_n-\mathbb{E}_{\hat{\fth}_n})(g_{x,t,D})=\sqrt{n}(\mathbb{P}_n-\mathbb{E}_{\fth_0})(g_{x,t,D})-\sqrt{n}(\mathbb{E}_{\hat{\fth}_n}-\mathbb{E}_{\fth_0})(g_{x,t,D})\nonumber\\
		&=\sqrt{n}(\mathbb{P}_n-\mathbb{E}_{\fth_0})(g_{x,t,D})- \sqrt{n}(\hat{\fth}_n-\fth_0)^{\textup{T}}\cdot \dot{\mathbb{E}}_{\fth_0}(g_{x,t,D})+o_{\mathbb{P}_{\fth_0}}(1).
	\end{align*}
	Thereafter, the Fréchet differentiability of the function $\fth \mapsto \mathbb{E}_{\fth}$ from Corollary \ref{cor1} and the asymptotic linearity of the estimation sequence $\sqrt{n}(\hat{\fth}_n-\fth_0)$ from Assumptions \ref{A4} and \ref{A5} are used. Rewriting the first summand in \eqref{e8} leads to:
	\begin{align}\label{e7}
		\begin{split}
			&\sqrt{n}\big(\mathbb{P}_n(g_{x,t,D})-\mathbb{E}_{\fth_0}(g_{x,t,D})\big)-\sqrt{n}\big(\mathbb{P}_n(\phi_{\fth_0})-\mathbb{E}_{\fth_0}(\phi_{\fth_0})\big)^{\textup{T}} \cdot \dot{\mathbb{E}}_{\fth_0}(g_{x,t,D})\\
			&+\Big[\sqrt{n}\big(\mathbb{P}_n(\phi_{\fth_0})-\mathbb{E}_{\fth_0}(\phi_{\fth_0})\big)^{\textup{T}} \cdot \frac{\dot{\alpha}_{\fth_0}}{2\sqrt{\alpha_{\fth_0}}}-\sqrt{n}\left(\sqrt{\mathbb{P}_n(\mathds{1}_D)}
			-\sqrt{\mathbb{E}_{\fth_0}(\mathds{1}_D)}\right)\Big]\\
			&\quad\cdot \left[2\, \mathbb{E}_{\fth_0}(g_{x,t,D})/\sqrt{\alpha_{\fth_0}}\right]+o_{\mathbb{P}_{\fth_0}}(1).
		\end{split}
	\end{align}
	The objective is now to derive the limiting distribution of the process delineated in \eqref{e7}. Corollary \ref{cor2} states that $\mathcal{G}_D$ is a Donsker class and, with Assumption \ref{A5}, it is $\Vert \phi_{\fth_0}\Vert_{\infty} < \infty$. By means of elementary considerations about the bracketing numbers, the finite class $\{\phi_{\fth_0}\}$ is Donsker too. As demonstrated in \citet[][Example 2.10.7]{vawell1996}, the Donsker property of two classes extends to their union with $\mathcal{G}_D^{\phi}:=\mathcal{G}_D \cup \{\phi_{\fth_0}\}$. 
	Accordingly, the sequence $\{\sqrt{n}(\mathbb{P}_n-\mathbb{E}_{\fth_0})(g_{x,t,D}): \, g_{x,t,D} \in\mathcal{G}_D^{\phi}\}$ converges in distribution in $\ell^{\infty}(\mathcal{G}_D^{\phi})$ to a $\mathbb{P}_{\fth_0}$-Brownian bridge $\mathbb{B}_{\mathbb{P}_{\fth_0}}$. The map
	\begin{equation*}
		\varphi: \left(\sqrt{n}(\mathbb{P}_n-\mathbb{E}_{\fth_0})\right) \mapsto 
		\begin{pmatrix}
			\sqrt{n}\big(\mathbb{P}_n-\mathbb{E}_{\fth_0}\big) \\
			\sqrt{n}\big(\mathbb{P}_n(\mathds{1}_D)-\mathbb{E}_{\fth_0}(\mathds{1}_D)\big)\\
			\sqrt{n}\big(\mathbb{P}_n(\phi_{\fth_0,1})-\mathbb{E}_{\fth_0}(\phi_{\fth_0,1})\big)\\
			\sqrt{n}\big(\mathbb{P}_n(\phi_{\fth_0,2})-\mathbb{E}_{\fth_0}(\phi_{\fth_0,2})\big)
		\end{pmatrix}
	\end{equation*}
	is continuous for all $(\sqrt{n}(\mathbb{P}_n-\mathbb{E}_{\fth_0})) \in \ell^{\infty}(\mathcal{G}_D^{\phi})$. Therefore, the continuous mapping theorem can be applied to establish the convergence in distribution of $\{\varphi\big(\sqrt{n}(\mathbb{P}_n-\mathbb{E}_{\fth_0})\big)(g_{x,t,D}): \, g_{x,t,D} \in\mathcal{G}_D^{\phi}\}$ in $\ell^{\infty}(\mathcal{G}_D^{\phi})$ to $\varphi(\mathbb{B}_{\mathbb{P}_{\fth_0}})$, i.e. $\big(\mathbb{B}_{\mathbb{P}_{\fth_0}},\mathbb{B}_{\mathbb{P}_{\fth_0}}(\mathds{1}_D),\mathbb{B}_{\mathbb{P}_{\fth_0}}(\phi_{\fth_0,1}),\mathbb{B}_{\mathbb{P}_{\fth_0}}(\phi_{\fth_0,2})\big)^{\textup{T}}$,
 	as $n \to \infty$. A comparison of the function $\varphi$ with equation \eqref{e7} reveals that the incorporation of the missing transformations by the root function necessitates the utilization of the delta method. Here, smoothness uses the Hadamard derivative, introduced in Appendix \ref{App2}. For this purpose, consider the function
	\begin{equation*}
		\phi_{\Delta}:( \ell^{\infty}(\mathcal{G}_D^{\phi}))^4 \to (\ell^{\infty}(\mathcal{G}_D^{\phi}))^4 , \quad \phi_{\Delta}(x,y,z_1,z_2)=(x,\sqrt{y},z_1,z_2)^{\textup{T}},
	\end{equation*}
	which is Hadamard differentiable according to Lemma \ref{l3}, correspondingly extended to four dimensions, with $\phi'_{\Delta}(x,y,z_1,z_2)=\textup{diag}(1,(2\sqrt{y})^{-1},1,1)$.
	Given the convergence in distribution $\varphi(\sqrt{n}(\mathbb{P}_n-\mathbb{E}_{\fth_0})) \to \varphi(\mathbb{B}_{\mathbb{P}_{\fth_0}})$ in $\ell^{\infty}(\mathcal{G}_D^{\phi})$ and the Hadamard differentiability of $\phi_{\Delta}$, we can employ the functional delta method \citep[][Theorem 20.8]{vaart1998}. Using this method allows us to account for the transformation by the root function; however, our limiting distribution will be affected. It yields the sequence
	\begin{align*}
		&\sqrt{n}\left(	\phi_{\Delta}
		\begin{pmatrix}
			\mathbb{P}_n \\
			\mathbb{P}_n(\mathds{1}_D)\\
			\mathbb{P}_n(\phi_{\fth_0,1})\\
			\mathbb{P}_n(\phi_{\fth_0,2})
		\end{pmatrix}
		-\phi_{\Delta}
		\begin{pmatrix}
			\mathbb{E}_{\fth_0} \\
			\mathbb{E}_{\fth_0}(\mathds{1}_D)\\
			\mathbb{E}_{\fth_0}(\phi_{\fth_0,1})\\
			\mathbb{E}_{\fth_0}(\phi_{\fth_0,2})
		\end{pmatrix}
		\right)
		= 
		\begin{pmatrix}
			\sqrt{n}\big(\mathbb{P}_n-\mathbb{E}_{\fth_0}\big) \\
			\sqrt{n}\left(\sqrt{\mathbb{P}_n(\mathds{1}_D)}-\sqrt{\mathbb{E}_{\fth_0}(\mathds{1}_D)}\,\right)\\
			\sqrt{n}\big(\mathbb{P}_n(\phi_{\fth_0,1})-\mathbb{E}_{\fth_0}(\phi_{\fth_0,1})\big)\\
			\sqrt{n}\big(\mathbb{P}_n(\phi_{\fth_0,2})-\mathbb{E}_{\fth_0}(\phi_{\fth_0,2})\big)
		\end{pmatrix}\\ 
		&\to
		\phi_{\Delta}'
		\begin{pmatrix}
			\mathbb{E}_{\fth_0} \\
			\mathbb{E}_{\fth_0}(\mathds{1}_D)\\
			\mathbb{E}_{\fth_0}(\phi_{\fth_0,1})\\
			\mathbb{E}_{\fth_0}(\phi_{\fth_0,2})
		\end{pmatrix} 
		\cdot 
		\begin{pmatrix}
			\mathbb{B}_{\mathbb{P}_{\fth_0}} \\
			\mathbb{B}_{\mathbb{P}_{\fth_0}}(\mathds{1}_D)\\
			\mathbb{B}_{\mathbb{P}_{\fth_0}}(\phi_{\fth_0,1})\\
			\mathbb{B}_{\mathbb{P}_{\fth_0}}(\phi_{\fth_0,2})
		\end{pmatrix} 
		= 
		\begin{pmatrix}
			\mathbb{B}_{\mathbb{P}_{\fth_0}} \\
			\mathbb{B}_{\mathbb{P}_{\fth_0}}(\mathds{1}_D)/(2\sqrt{\alpha_{\fth_0}})\\
			\mathbb{B}_{\mathbb{P}_{\fth_0}}(\phi_{\fth_0,1})\\
			\mathbb{B}_{\mathbb{P}_{\fth_0}}(\phi_{\fth_0,2})
		\end{pmatrix}
	\end{align*}
	which converges in distribution in $\ell^{\infty}(\mathcal{G}_D^{\phi})^4$ as $n \to \infty$. Finally, the four components of the vector can be combined using the continuous mapping theorem from \citet[][Theorem 18.11]{vaart1998}, which results in the expression given in \eqref{e7}. Under $H_0$, the sequence
	\begin{align*}
		&\Big\{\sqrt{n}\big(\mathbb{P}_n-\mathbb{E}_{\fth_0}\big)-\sqrt{n}\big(\mathbb{P}_n(\phi_{\fth_0})-\mathbb{E}_{\fth_0}(\phi_{\fth_0})\big)^{\textup{T}} \cdot \dot{\mathbb{E}}_{\fth_0}\\
		&+\Big[\sqrt{n}\big(\mathbb{P}_n(\phi_{\fth_0})-\mathbb{E}_{\fth_0}(\phi_{\fth_0})\big)^{\textup{T}} \cdot \frac{\dot{\alpha}_{\fth_0}}{2\sqrt{\alpha_{\fth_0}}}-\sqrt{n}\left(\sqrt{\mathbb{P}_n(\mathds{1}_D)}
		-\sqrt{\mathbb{E}_{\fth_0}(\mathds{1}_D)}\right)\Big]\\
		&\quad\cdot \left[2\, \mathbb{E}_{\fth_0}/\sqrt{\alpha_{\fth_0}}\right](g_{x,t,D}): g_{x,t,D} \in \mathcal{G}_D^{\phi}\Big\}	
	\end{align*}
	converges in distribution in $\ell^{\infty}(\mathcal{G}_D^{\phi})$ and therefore also in $\ell^{\infty}(\mathcal{G}_D)$ \cite[see][Theorem 19.23]{vaart1998} to the Gaussian process
	\begin{align*}
		\mathbb{H}_{\mathbb{P}_{\fth_0}}(g_{x,t,D})= & \, \mathbb{B}_{\mathbb{P}_{\fth_0}} (g_{x,t,D})-\left(\mathbb{B}_{\mathbb{P}_{\fth_0}}(\phi_{\fth_0,1}),\mathbb{B}_{\mathbb{P}_{\fth_0}}(\phi_{\fth_0,2})\right) \cdot \dot{\mathbb{E}}_{\fth_0}(g_{x,t,D})\\
		&+\left[\left(\mathbb{B}_{\mathbb{P}_{\fth_0}}( \phi_{\fth_0,1}), \mathbb{B}_{\mathbb{P}_{\fth_0}}( \phi_{\fth_0,2})\right)  \cdot \dot{\alpha}_{\fth_0}
		- \mathbb{B}_{\mathbb{P}_{\fth_0}}(\mathds{1}_D) \right]\cdot \frac{\mathbb{E}_{\fth_0}(g_{x,t,D})}{\alpha_{\fth_0}}
	\end{align*}
	for $g_{x,t,D} \in \mathcal{G}_D$. \qed
\end{proof}

Table \ref{tabgenesis} lists the distribution and also compares it to the situation of known $\fth_0$ and estimated $n$ (with derivation suppressed here) in order to show the increasing complexity.

The corresponding convergence of the sequence in the supremum norm $\Vert \cdot \Vert_{\mathcal{G}_D}$, i.e. the convergence of the modified Kolmogorov-Smirnov test statistic \eqref{e11} (or equivalently \eqref{e10}), follows from the continuous mapping theorem.


\subsection{Computation of critical value} \label{sec4_4}

The distribution of $\mathbb{H}_{\mathbb{P}_{\fth_0}}$ has been derived under hypothesis in Theorem \ref{th1}. We now determine that of $\Vert \mathbb{H}_{\mathbb{P}_{\fth_0}} \Vert_{\mathcal{G}_D}$, which can be simulated using the involved Gaussian processes. The Gaussian process $\mathbb{H}_{\mathbb{P}_{\fth_0}}$ has zero mean and Appendix \ref{App3} derives as a general representation of the covariance between $\mathbb{H}_{\mathbb{P}_{\fth_0}}(g_{x_1,t_1,D})$ and $\mathbb{H}_{\mathbb{P}_{\fth_0}}(g_{x_2,t_2,D})$: 
\begin{align}\label{e9}
	\begin{split}
	&\mathbb{E}_{\fth_0}(g_{x_1,t_1,D}g_{x_2,t_2,D})-\mathbb{E}_{\fth_0}(g_{x_1,t_1,D})\mathbb{E}_{\fth_0}(g_{x_2,t_2,D})/\alpha_{\fth_0}+K_2^{\textup{T}} \mathbb{E}_{\fth_0}(g_{x_1,t_1,D}\phi_{\fth_0})\\
	&\quad+K_1^{\textup{T}}\mathbb{E}_{\fth_0}(g_{x_2,t_2,D} \phi_{\fth_0})+K_1^{\textup{T}} \mathbb{E}_{\fth_0}(\phi_{\fth_0}\phi_{\fth_0}^{\textup{T}})K_2,
	\end{split}
\end{align}
where $K_i=(K_{i,\theta_0}, K_{i, \vth_0})^{\textup{T}}$ and $K_{i,j}:=\left(\mathbb{E}_{\fth_0}g_{x_i,t_i,D}\frac{\dot{\alpha}_{j}}{\alpha_{\fth_0}}-\dot{\mathbb{E}}_{j}g_{x_i,t_i,D}\right)$, for $i=1,2$ and $j \in \{\theta_0, \vth_0\}$. (Note that $i$ and $j$ here are not indicators for units as in Section \ref{popmodass}.)  For the models of Section \ref{secparmod}, more specific representations will be provided. Therefore, the expectations of the indicator function $g_{x_i,t_i,D}$ and of the score function $\phi_{\fth_0}$, as well as the derivative $\dot{\mathbb{E}}_{j}$ need to be calculated. Those will be given in closed form in Appendix \ref{App4}. (For calculating $\mathbb{E}_{\theta_0}(g_{\min(x_1, x_2),\min(t_1,t_2),D})$, we simply need to determine the minimum values and put it into $\mathbb{E}_{\theta_0}(g_{x_i,t_i,D})$.) As usual for practical applications, $\fth_0$ is replaced by $\hat{\fth}_n$ by use of the continuous mapping theorem. The commonly used Cholesky decomposition for random field generation \cite[see e.g.][]{liu2019} will be simplified when applied to simulate the distribution of $\Vert \mathbb{H}_{\mathbb{P}_{\hat{\fth}_n}} \Vert_{\mathcal{G}_D}$, because the support is compact.

\rule{\linewidth}{0.5pt}
\hypertarget{h8}{\textbf{Algorithm 2:}} Computing the critical value of \eqref{e10} \vspace{-3mm}\\
\rule{\linewidth}{0.5pt}
\begin{enumerate}
	\setlength{\itemsep}{-6pt}
	\item Create a grid on the rectangle $[0, G+s]\times [0,G]$.
	\item Calculate the covariance matrix for all grid points according to \eqref{e9}.
	\item Repeat for a specified number of iterations:
	\begin{enumerate}
		\setlength{\itemsep}{-6pt}
		\renewcommand{\labelenumii}{\roman{enumii})}
		\item Create a normally distributed vector with mean zero and the calculated covariance matrix.
		\item Find the maximum absolute value among the components of the vector.
	\end{enumerate}
	\item Select the appropriate quantile from the sorted maximum values. \vspace{-5mm}
\end{enumerate}
\rule{\linewidth}{0.5pt}



\section{Testing the exponential distribution}\label{secparmod}

The main applied incentive was to test whether the  lifespan $X$ has an exponential density $f_{\theta}^X(x)=\theta exp(- \theta x)$, as part of hypothesis \eqref{hypoth}. Typically in practice, we cannot assume that lifespan $X$ and the age at truncation $T$ are stochastically independent. Mainly for the ease of display, we assume $T$ to be uniformly distributed. 

We test with dependence, but as a reference we also include the test without dependence.  We verify the validity of the assumptions \ref{A1}-\ref{A5} for two specific models and derive the covariance function needed for the Algorithm \protect\hyperlink{h8}{2}.  

\subsection{Independent truncation: Product copula} \label{sec61}

The goodness-of-fit test for the hypotheses \eqref{hypoth} assuming independence between the lifetime and truncation variable, i.e. with product copula $C_{\Pi}(u,v):=u \cdot v$, is given by Theorem \ref{th1} and practically enabled with Algorithms 1 and 2. We need to verify Assumptions \ref{A1}-\ref{A5}. Note that the copula does not depend on a parameter, so that under hypothesis $F_{\theta}=F^X_{\theta} \cdot F^T$ with $\Theta=(\varepsilon, 1/\varepsilon)$. 

Assumption \ref{A1} results from \citet[][Eq. (1)]{weiswied2021} and $0 < G < \infty$, $s>0$, so that 
\begin{equation*}
	\alpha_{\theta}=\frac{(1-e^{- \theta s})(1-e^{-G \theta})}{G \theta} \in (0,1).
\end{equation*} 
The integrability and continuous differentiability of the density $f_{\theta}(x,t)=\theta e^{-\theta x}/G$ are consequences of the corresponding properties of the exponential function. The interchangeability of integrating for $(x,t)$ and differentiating for $\theta$ can be demonstrated through direct calculations. The $Z$-estimator is defined as the zero of the criterion function given in \citet[][Eq. (5)]{weiswied2021}, where the score function is defined as
\begin{equation}\label{e15}
	\psi_{\theta}(X_i,T_i)=\mathds{1}_D(X_i,T_i)\left[X_i-\frac{1}{\theta}+\frac{\dot{\alpha}_{\theta}}{\alpha_{\theta}}\right].
\end{equation}
\citet[][Theorems 1 and 2]{weiswied2021} show that Assumptions \ref{A3} and \ref{A4} hold. Note that the relation $\phi_{\theta_0}=-(\mathbb{E}_{\theta_0} [ \dot{\psi}_{\theta_0}])^{-1}\psi_{\theta_0}$ and Assumption \ref{A5} follows from \citet[][Lemmas 2 and 3]{weiswied2021}. Especially, it is $\sup_{(x,t) \in D} \vert x - \frac{1}{\theta_0}+\frac{\dot{\alpha}_{\theta_0}}{\alpha_{\theta_0}} \vert < \infty$ and $\dot{\mathbb{E}}_{\theta_0}[\psi_{\theta_0}]>0$ and thus $\sup_{(x,t) \in S} \vert \phi_{\theta_0}(x,t)\vert < \infty$. Therefore, Theorem \ref{th1} applies and due to the information matrix equality \cite[see][Eq. (11)]{topweis2025} and the relation $\psi_{\theta_0}=\dot{f}_{\theta_0}/f_{\theta_0}-\dot{\alpha}_{\theta_0}/\alpha_{\theta_0}$ on $D$, the covariance function simplifies to
\begin{align}\label{e51}
	\begin{split}
		&\mathbb{E}_{\theta_0}(g_{\min(x_1,x_2),\min(t_1,t_2),D})- \mathbb{E}_{\theta_0}(g_{x_1,t_1,D}) \mathbb{E}_{\theta_0}(g_{x_2,t_2,D})/\alpha_{\theta_0}\\
		&+(\mathbb{E}_{\theta_0} [ \dot{\psi}_{\theta_0}])^{-1} \cdot \mathbb{\dot{E}}_{\theta_0}(g_{x_1,t_1,D}) \mathbb{\dot{E}}_{\theta_0}(g_{x_2,t_2,D})\\
		&-(\mathbb{E}_{\theta_0} [ \dot{\psi}_{\theta_0}])^{-1} \cdot  \frac{\dot{\alpha}_{\theta_0}}{\alpha_{\theta_0}}\cdot \left[\mathbb{\dot{E}}_{\theta_0}(g_{x_2,t_2,D}) \mathbb{E}_{\theta_0}(g_{x_1,t_1,D})+\mathbb{\dot{E}}_{\theta_0}(g_{x_1,t_1,D}) \mathbb{E}_{\theta_0}(g_{x_2,t_2,D})\right]\\
		&+(\mathbb{E}_{\theta_0} [ \dot{\psi}_{\theta_0}])^{-1} \cdot  \frac{\dot{\alpha}_{\theta_0}^2}{\alpha_{\theta_0}^2}\cdot \mathbb{E}_{\theta_0}(g_{x_1,t_1,D})\mathbb{E}_{\theta_0}(g_{x_2,t_2,D}).
	\end{split}
\end{align}

The elements are given in Appendix \ref{App41}. In order to differentiate the symbols, we use the superscript $\Pi$ in $C$, $\hat{\theta}_n$ and $\alpha_{\hat{\theta}_n}$ as a signal of the involved product copula in Section \ref{emp_appl} of the application.   

\subsection{Dependent truncation: $\textup{FGM}$ copula} \label{sec62}

 \cite{topweis2025} apply the Gumbel-Barnett copula for modelling dependence between lifespan and truncation age, which is only suitable for human demography, as the direction of instationary is known to be ``upwards''. By contrast, for business demography, where the direction is not known, the Farlie-Gumbel-Morgenstern (FGM) copula, which models two-sided dependence, is more appropriate, especially for the data example in Section \ref{emp_appl}: 
\begin{equation*}
	C^{\textup{FGM}}_{\vth}(u,v):=u v  [1+\vth (1-u)(1-v)], \, \vth \in [-1,1]
\end{equation*}
In order to distinguish between the symbols in Section \ref{emp_appl} of the application, we already use the superscript $\textup{FGM}$ in this section. 
Kendall's tau is $\tau_{\vth}^{\textup{FGM}}=2\vth/9$, so that the codomain of the parameter $\vth$ enables positive $(\vth>0)$ and negative dependence $(\vth<0)$, as well as independence $(\vth=0)$. The density of $(X_i,T_i)^{\textup{T}}$ results in
\begin{equation*}
	f^{\textup{FGM}}_{\fth }(x,t)=
	\frac{\theta}{G}e^{-\theta x} \left[1+\vth (2e^{-\theta x} -1) \left(1-\frac{2t}{G}\right)\right],
\end{equation*}
for $x>0$ and $0<t<G$, and is zero elsewhere. Hence, the parameter space is two-dimensional with $\Theta=(\varepsilon_{\theta}, 1/\varepsilon_{\theta})\times(\varepsilon_{\vth}-1,1-\varepsilon_{\vth})$. For the goodness-of-fit test of the hypotheses \eqref{hypoth}, it is now $F_{\fth}=C^{\textup{FGM}}_{\vth}(F^X_{\theta},F^T)$. Again, we need to verify the Assumptions \ref{A1}-\ref{A5}, now for the FGM copula. 

For the Gumbel-Barnett copula, Assumptions \ref{A1}-\ref{A5} are proven in \citet[][Lemmas 1,3,4, Theorem 1,2]{topweis2025}. 
That is, in order to prove Asumption \ref{A1} for the FGM copula, the selection probability calculates as \cite[see][after Eq. (2)]{topweis2025}:
\begin{align*}
	\alpha_{\fth}^{\textup{FGM}}
	&=\frac{1}{\theta G} (1-e^{-\theta s})(1-e^{-\theta G})\\
	&\quad-\frac{\vth}{G}\left[-\frac{1}{\theta}(e^{-\theta s}-1)(e^{-\theta G}+1)+\frac{1}{2 \theta} (e^{-2\theta s}-1)(e^{-2\theta G}+1)\right]\\
	&\quad +\frac{2 \vth}{G^2} \left[\frac{1}{\theta^2} (1-e^{-\theta s})(1-e^{-\theta G})-\frac{1}{4 \theta^2} (1-e^{-2\theta s})(1-e^{-2\theta G})\right].
\end{align*}
and it can proven to be in $(0,1)$. 

In order to see Assumption \ref{A2}, note that the integrability and continuous differentiability of the density $f_{\fth}^{\textup{FGM}}$ are consequences of the corresponding properties of the exponential function. The interchangeability of integrating for $(x,t)$ and differentiating for $\theta$ can be demonstrated through straightforward calculations.

The consistency stated in Assumption \ref{A3} can be established using techniques analogous to those applied in the Gumbel–Barnett copula case in \citet[][Theorem 1]{topweis2025}, starting from the two-dimensional score function $\psi_{\fth}^{\textup{FGM}}$ with coordinates:
\begin{equation}\label{e19}
	\begin{split}
	\psi_{\fth,1}^{\textup{FGM}}(X_i,T_i)=\mathds{1}_{[T_i,T_i+s]}(X_i)\Bigg[\frac{1}{\theta}-X_i-\frac{2\vth X_i e^{-\theta X_i}\left(1-\frac{2T_i}{G}\right)}{1+\vth \left(2e^{-\theta X_i}-1\right)\left(1-\frac{2T_i}{G}\right)}
	-\frac{\dot{\alpha}_{\theta}^{\textup{FGM}}}{\alpha_{\fth}^{\textup{FGM}}}\Bigg] \\
	\psi_{\fth,2}^{\textup{FGM}}(X_i,T_i)=\mathds{1}_{[T_i,T_i+s]}(X_i)\left[\frac{\left(2 e^{-\theta X_i}-1\right)\left(1-\frac{2T_i}{G}\right)}{1+\vth \left(2e^{-\theta X_i}-1\right)\left(1-\frac{2T_i}{G}\right)}-\frac{\dot{\alpha}_{\vth}^{\textup{FGM}}}{\alpha_{\fth}^{\textup{FGM}}}\right]
		\end{split}
\end{equation}
 Note that the FGM copula has a simpler structure than the Gumbel-Barnett copula and simplifies calculations. Assumption \ref{A4}, the asymptotic linearity with the relation $\phi_{\fth_0}=-(\mathbb{E}_{\fth_0} [ \dot{\psi}_{\fth_0}])^{-1}\psi_{\fth_0}$, follows similarly as in \citet[][Theorem 2]{topweis2025}. On proving Assumption \ref{A5} see \citet[][Lemmas 3 and 4]{topweis2025}. Especially, for verifying $\sup_{(x,t)\in S}\Vert \phi_{\fth_0}(x,t)\Vert <\infty$, note that for the denominators in $\psi_{\fth_0,1}$ and $\psi_{\fth_0,2}$ it is $$\left\vert 1 + \vth_0 (2e^{-\theta_0 x}-1) \left(1-\frac{2t}{G}\right)\right\vert\geq \left\vert 1 - \left\vert \vth_0 (2e^{-\theta_0 x}-1) \left(1-\frac{2t}{G}\right)\right\vert \right\vert \geq 1- \vert \vth_0 \vert$$ on $D$ by the reverse triangle inequality. The edge case $\vert\vth_0\vert=1$ is not included in the parameter space. The two longer numerators in \eqref{e19} can be bounded above on $D$ by
 $\left\vert 2 \vth_0 x e^{-\theta_0 x} \left(1-\frac{2t}{G}\right) \right\vert \leq 2 \vert \vth_0\vert (G+s)$ and $\left\vert (2 e^{-\theta_0 x} -1) \left(1-\frac{2t}{G}\right) \right\vert \leq 1$.

Thus, Theorem \ref{th1} applies, and due to the information matrix equality \cite[see][Eq. (11)]{topweis2025} and the relation $\psi_{\fth_0}=\dot{f}_{\fth_0}/f_{\fth_0}-\dot{\alpha}_{\fth_0}/\alpha_{\fth_0}$ on $D$, we can simplify the covariance function to
\begin{align}\label{e18}
	\begin{split}
		&\mathbb{E}_{\fth_0}(g_{\min(x_1,x_2),\min(t_1,t_2),D})- \mathbb{E}_{\fth_0}(g_{x_1,t_1,D}) \mathbb{E}_{\fth_0}(g_{x_2,t_2,D})/\alpha_{\fth_0}\\
		&-\mathbb{\dot{E}}_{\fth_0}(g_{x_2,t_2,D})^{\textup{T}}\cdot (-\mathbb{E}_{\fth_0} [ \dot{\psi}_{\fth_0}])^{-1} \cdot \mathbb{\dot{E}}_{\fth_0}(g_{x_1,t_1,D}) \\
		&+\frac{\mathbb{E}_{\fth_0}(g_{x_1,t_1,D})}{\alpha_{\fth_0}} \cdot \mathbb{\dot{E}}_{\fth_0}(g_{x_2,t_2,D})^{\textup{T}}\cdot  (-\mathbb{E}_{\fth_0} [ \dot{\psi}_{\fth_0}])^{-1} \cdot  \dot{\alpha}_{\fth_0}\\
		&+\frac{\mathbb{E}_{\fth_0}(g_{x_2,t_2,D})}{\alpha_{\fth_0}} \cdot \mathbb{\dot{E}}_{\fth_0}(g_{x_1,t_1,D})^{\textup{T}}\cdot  (-\mathbb{E}_{\fth_0} [ \dot{\psi}_{\fth_0}])^{-1} \cdot  \dot{\alpha}_{\fth_0}\\
		&-\frac{\mathbb{E}_{\fth_0}(g_{x_1,t_1,D})}{\alpha_{\fth_0}}\frac{\mathbb{E}_{\fth_0}(g_{x_2,t_2,D})}{\alpha_{\fth_0}} \cdot \dot{\alpha}_{\fth_0}^{\textup{T}} \cdot  (-\mathbb{E}_{\fth_0} [ \dot{\psi}_{\fth_0}])^{-1} \cdot  \dot{\alpha}_{\fth_0}.
	\end{split}
\end{align}

The elements are given in Appendix \ref{App42}. The superscript $\textup{FGM}$ will indicate the FGM copula in Section \ref{emp_appl}.


\subsection{Example: Enterprise lifespans in Germany}\label{emp_appl}
We wanted to conduct a goodness-of-fit test to determine whether enterprise lifespans in Germany can be described by an exponential distribution.  

\subsubsection{Population and observed data} \label{obsdata}

The population comprises enterprises based in Germany and founded between January 1, 1990, and December 31, 2013 ($G = 24$). The age of the enterprise $i$ at the end of the foundation period $T_i$ is in $[0,24]$, the support of the uniform distribution. 
 The variable of interest, $X_i>0$, is the age of enterprise $i$ at closure. In the three years, from January 1, 2014, to December 31, 2016 ($s = 3$) enterprise closures are observed, together with their foundation dates, yielding a doubly truncated sample of $m_n = 55{,}279$ units. Unfortunately, observed enterprise closures are restricted to insolvencies. For simplicity, and in order to report the life expectancy of enterprises, we assume in a competing risk model for the enterprise history, that the intensity of closure due to insolvency is equal to the intensity due to any other reason. 


\subsubsection{Results for $\boldsymbol{\Pi}$ and $\textup{FGM}$ copula}

The resulting tests for the models in Sections \ref{sec61} and \ref{sec62}, decomposed into test statistic and critical value:

We compute the two test statistics with Algorithm \hyperlink{h4}{1}. For the observations ($m_n=55,279$), the projections onto the upper (15,268) and right edge (6,768), as well as the intersection points (5,971,513), are determined. A total of 6,048,849 evaluations are therefore required. For $F^{\textup{obs}}_{\hat{\boldsymbol{\theta}}_n}$ and $\alpha_{\hat{\boldsymbol{\theta}}_n}$, the point estimate are needed. 
Assuming the age of the enterprise to be independent of its foundation year, or, of its age at the time of study onset, is the situation in  \cite{weiswied2021} and results in  $\hat{\theta}_n^{\Pi}=0.08261$. (As a consequence, the observation probability is $\alpha_{\hat{\theta}_n}^{\Pi}=0.0955$.) When we assume that the dependence between the age of the company and its age at the study beginning is modelled by the FGM copula,  \cite{topweis2025} reported $\hat{\theta}^{\textup{FGM}}_n=0.08172$ and $\hat{\vth}^{\textup{FGM}}_n=0.10256$. Note that the positive dependence between $X_i$ and $T_i$ means a negative time trend for life expectancy. Kendall's tau seems small with $\tau_{\hat{\vth}_n}^{\textup{FGM}}=0.023$, but \cite{topweis2025} show significance. The observation probability is $\alpha_{\hat{\fth}_n}^{\textup{FGM}}=0.09753$. Auxiliary values in the Algorithm \hyperlink{h4}{1} are now given in Table \ref{table3}.

\vspace{-1em}
\begin{table}[H]
	\caption{Auxiliary results for KS statistic according to Algorithm \protect\hyperlink{h4}{1} for the independence $C^{\Pi}$ and FGM copula $C^{\textup{FGM}}_{\vth}$, $\times 10^2$.}	\label{table3} 
	\begin{center}
		\begin{tabular}{cccccc}
			\toprule
			&			$\delta^1$ & $\delta^2$ & $\delta^3$ 				 & $\delta^4$ 					& $\delta^5$ \\ \midrule[1.3pt]	
			$C^{\Pi}$ & 	$2.2725$  & $2.2733$  & $3.6173 \cdot 10^{-3}$ & $1.7721 \cdot 10^{-3}$ & $1.7728 \cdot 10^{-3}$\\
			\bottomrule[1.3pt]	
			$C^{\textup{FGM}}_{\vth}$ &		$2.3040$  & $2.3039$  & $3.6173 \cdot 10^{-3}$ & $1.7697 \cdot 10^{-3}$ & $1.7702 \cdot 10^{-3}$\\
			\bottomrule[1.3pt]			
		\end{tabular}
	\end{center}
\end{table}
\vspace{-1em}
Thereafter, the largest deviation occurs for the copulas in $\delta^1$ or $\delta^2$, that is, in the observations or the intersection points. By multiplying the maximum by $\sqrt{m_n}$ and $\sqrt{\alpha_{\hat{\theta}_n}^{\Pi}}$ or $\sqrt{\alpha_{\hat{\fth}_n}^{\textup{FGM}}}$, the test statistics can be calculated as 1.651 and 1.6917, respectively.

We continue with the calculation of the quantile under the null hypothesis, using the method described in Algorithm \hyperlink{h8}{2}. A grid width of $1/100=0.01$ was chosen, and the number of entries of the covariance matrix is already $(100G\cdot 100(G+s) )^2=4.199 \cdot 10^{13}$. The procedure was repeated 1,000 times, and for instance at level $99$\%, quantiles of 0.3558 for independence, and of 0.2914 for FGM-dependence, have been determined. Table  \ref{table4} contains quantiles at levels $90$\%, $95$\% and  $99$\%, and compares them with the test statistics. 
\begin{table}[H]
	\caption{Comparison of the test statistics \eqref{e10} (according to Algorithm \protect\hyperlink{h8}{1}) and critical values according to Algorithm \protect\hyperlink{h8}{2} (three nominal levels) for $\Pi$ copula (with the covariances \eqref{e51}) and for $\textup{FGM}$ copula (with the covariances \eqref{e18})}
	\label{table4}
		\begin{tabular}{c|ccc|ccc}
			
			\toprule
			Estimators$\backslash$ Level $\alpha$
			& 0.10 & 0.05 &0.01 & 0.10 & 0.05 &0.01 \\ \midrule[1.3pt]
			& \multicolumn{3}{c|}{$C^{\Pi}$: \eqref{e10} \, = \, 1.65} & \multicolumn{3}{c}{$C^{\textup{FGM}}_{\vth}$: \eqref{e10} \, = \, 1.69} \\
			$(\hat{\theta}_n,n)$  & $0.5719$  & $0.6558$ & $0.8298$ & $0.7320$  & $0.8510$ & $1.0816$ \\
			$(\hat{\theta}_n,\hat{n})$  & $0.2869$  & $0.3051$ & $0.3558$ & $0.2378$  & $0.2535$ & $0.2914$\\
			\bottomrule[1.3pt]	
		\end{tabular}
\end{table}	
The test rejects clearly at any of those levels. The rejected independence of  \cite{topweis2025}, suggest that the first test is of little use. The second test still does not strictly allow rejecting the exponential distribution of the lifespan, because the uniform distribution of the foundation date may also have lead to rejection. However,  \citet[][Lemma 2]{doerre2020} shows that the assumption of a homogeneous Poisson process as a model for foundation (with whatever parameter) results in a uniform distribution of $T$, so that the exponential distribution assumption for $X$ is most likely the (or one) reason for rejection. (Table \ref{table4} also lists quantiles for $(\hat{\fth}_n, n)$ where it is itself interesting that the quantiles can be given in  abundance of the estimation of $n$. Of course, the test statistics remain the same, but the quantiles more than double, so that the knowledge of $n$ appears to do more harm than good.)

\section{Conclusion} \label{secdiscussion}

Theoretically, it is interesting to note that the only additional assumption above of a a parametric analysis  \cite[see][]{weiswied2021,topweis2025} is the first requirement in Assumption \ref{A5}, the finiteness for $\phi_{\boldsymbol{\theta}_0}$. The function $\phi_{\boldsymbol{\theta}_0}$ results from both the population model and the sampling design, so that it is not easy to assess whether the restriction is notable. A straightforward suggestion of a larger population model is from  \cite{weisdoer2022} who present evidence, but only under the assumption of truncation independence, that foundation does not evolve perfectly uniform. Hence, the choice of the marginal distribution of $T$ should be modelled.
 
Extensions of the sampling design may include, for instance, additional censoring, which is occasionally associated with truncation \cite[see e.g.][]{honore2012,siegtopweis2025}. An implication is that the population definition must be extended to cover the study period.  When we think of the design as an approximation of realistically observational data, there is another options to ours, namely to follow \cite{And0} and consider the data as a truncated sample \cite[and denote it as ``clockwise'' design as in][Figure 2]{topweis2025}. \cite{emura2017} and \cite{doerreemura2021} view the data as a sample of a truncated population (the ``anti-clockwise'' design). There are substantial implications: (i) The number of observations is random in the clockwise design and fixed in the anti-clockwise. (ii) In the clockwise design, forecasting is possible for those population units that are alive at the end of the study period. For the anti-clockwise design, the population ends with the study period,  so that forecasting is not possible. (iii) The independence between units, necessary to write the likelihood as a product, is more plausible for the clockwise design, where lives are temporally widely separated. Without much evidence being given here, we found that, at least in the special case of left-truncation and known $n$, our goodness-of-fit test in the clockwise sense, which is, under $H_0$, based on a Poisson approximation of the point process, behaves similarly to the goodness-of-fit test in the anti-clockwise design of \cite{emura2012}. 

In our large-sample application, the test statistic is largely separated from the critical value. For small-sample cases with a less marked test result, actual level and power could be of interest and assessed with detailed simulations. The partial derivatives of copulas used here are given in closed form, so that the conditional inverse method simplifies such simulations. 

Note also that Theorem \ref{th1} is valid for other global measures of discrepancy than the Kolmogorov-Smirnov concept, such as $\int \mathbb{G}_n(g)^2\dt \mathbb{P}_{\fth_0}$, which leads to the Cramér-von-Mises statistic. However, especially the computation of the test statistic must be adapted. Another computational aspect arises from the fact that we used in Algorithm 2, with the Cholesky decomposition, a very basic method. It has cubic complexity and is hence relatively slow, whereas \cite{liu2019} present an overview of several other methods with lower complexity, partly at the expense of accuracy. 

\vspace*{0.3cm}
\textbf{Acknowledgment}:
We thank Dr. Claudia Gregor-Lawrenz (Bundesanstalt für Finanzdienstleistungsaufsicht (BaFin), Bonn) for her continuous support on the topic, e.g. by contributing earlier data in the truncation design. We thank Dr. Wolfram Lohse (Görg, Hamburg) for indication to the data source and Daniel Ollrogge (CRIF Bürgel, Hamburg) for supply of the data. 
For helpful comments we thank the participants of the presentations at the sessions ``Survival analysis: Truncated data'' (CFE-CMStatistis, Dec. 2024, London, UK), ``Empirical Economics and Applied Econometrics 2'' (DAGStat, Mar. 2025, Berlin, Germany) and ``Copula Models in Econometrics and other Applications'' (EcoStat, Aug. 2025, Tokyo, Japan). The linguistic and idiomatic advice of Brian Bloch is also gratefully acknowledged.

 

\appendix
\renewcommand{\theequation}{\thesection.\arabic{equation}}
\numberwithin{equation}{section}
\newpage 
\section{Proof of Corollary \ref{cor1}, Lemma \ref{l2} and \ref{l1}}\label{App1}
\subsection{Proof of Corollary \ref{cor1}}\label{App11}
Since the set of functions $\mathcal{G}_D$ is limited to those of the form $\mathds{1}_{[0,x]\times[0,t]}\mathds{1}_D$ for $(x,t) \in S$, the Fréchet derivative can be specified explicitly. Using the density function $f_{\fth}$, the mapping $g_{x,t,D} \in \mathcal{G}_D \mapsto \mathbb{E}_{\fth_0}(g_{x,t,D}) \in \mathbb{R}$ can be written as
\begin{equation*}
	\mathbb{E}_{\fth_0}(g_{x,t,D})=\int_{[0,x]\times[0,t]\cap D} f_{\fth_0}(\tilde{x}, \tilde{t}) \dt (\tilde{x}, \tilde{t}).
\end{equation*}
According to \ref{A2}, the derivatives of $f_{\fth}$ with respect to $\theta$ and $\vth$ are bounded, i.e. $\Vert \dot{f}_{\fth}(x,t)\Vert \leq h_{\fth}$ for all $(x,t) \in S$ and $\fth \in \Theta$. Pursuant to \citet[][Lemma 16.2]{bauer2001}, the mapping $\fth \mapsto \mathbb{E}_{\fth}(g_{x,t,D})$ is continuously differentiable at $\fth_0$ given a function $g_{x,t,D} \in \mathcal{G}_D$. The (total) derivative at $\fth_0$ is given by
\begin{equation*}
	\dot{\mathbb{E}}_{\fth_0}(g_{x,t,D})=\int_{[0,x]\times[0,t]\cap D} \dot{f}_{\fth_0}(\tilde{x},\tilde{t}) \dt (\tilde{x}, \tilde{t}).
\end{equation*}
Fréchet differentiability follows from the differentiability of $f_{\fth}$ at $\fth_0$ and Scheffé's lemma, because
\begin{align*}
	&\lim_{\Vert h \Vert\to 0} \sup_{(x,t) \in S} \left\vert \int_{[0,x]\times[0,t]\cap D} \frac{f_{\fth_0+h}(\tilde{x}, \tilde{t})-f_{\fth_0}(\tilde{x}, \tilde{t})- h^{\textup{T}}\dot{f}_{\fth_0}(\tilde{x},\tilde{t})}{\Vert h\Vert} \dt (\tilde{x}, \tilde{t})\right\vert\\
	&\leq \lim_{\Vert h \Vert \to 0}   \int_{D} \left\vert\frac{f_{\fth_0+h}(\tilde{x}, \tilde{t})-f_{\fth_0}(\tilde{x}, \tilde{t})- h^{\textup{T}}\dot{f}_{\fth_0}(\tilde{x},\tilde{t})}{\Vert h\Vert}  \right\vert\dt (\tilde{x}, \tilde{t})=0.
\end{align*}
Consequently, $\dot{\mathbb{E}}_{\fth_0}$ is the Fréchet derivative of $\fth\in  [\varepsilon_{\theta}, 1/\varepsilon_{\theta}] \mapsto \mathbb{E}_{\fth} \in \ell^{\infty}(\mathcal{G}_D)$ at $\fth_0$. \qed

\subsection{Proof of Lemma \ref{l2}}\label{App12}
	First, note that according to \citet[][Lemma 16.2]{bauer2001}, it follows from \ref{A2} that $\fth \mapsto \alpha_{\fth}$ is continuously differentiable at $\fth_0$. Because of $\alpha_{\fth} \in (0,1)$ from \ref{A1}, the map $\fth \mapsto \sqrt{\alpha_{\fth}}$ is continuously differentiable at $\fth_0$ too. Hence, it is 
	\begin{equation*}
		\sqrt{\alpha_{\fth}}-\sqrt{\alpha_{\fth_0}}-(\fth-\fth_0)^{\textup{T}} \cdot \frac{1}{2\sqrt{\alpha_{\fth_0}}} \dot{\alpha}_{\fth_0}= o(\Vert \fth -\fth_0\Vert).
	\end{equation*}
	With the consistency of Assumption \ref{A3} and \citet[][Lemma 2.12]{vaart1998} the equation
	\begin{equation*}
		\sqrt{\alpha_{\hat{\fth}_n}}-\sqrt{\alpha_{\fth_0}}=(\hat{\fth}_n-\fth_0)^{\textup{T}} \cdot \frac{\dot{\alpha}_{\fth_0}}{2\sqrt{\alpha_{\fth_0}}} + o_{\mathbb{P}_{\fth_0}}(\Vert \hat{\fth}_n -\fth_0\Vert)
	\end{equation*}
	becomes evident also. The representation \eqref{e3} follows from the asymptotic linearity \ref{A4} with \ref{A5} of the sequence $\sqrt{n}(\hat{\fth}_n-\fth_0)$. Thereby, note that $o_{\mathbb{P}_{\fth_0}}(1)+o_{\mathbb{P}_{\fth_0}}(\Vert \hat{\fth}_n -\fth_0\Vert)=o_{\mathbb{P}_{\fth_0}}(1)$. \qed

\subsection{Proof of Lemma \ref{l1}}\label{App13}
	The convergence can be shown in two steps:
	\begin{equation}\label{e6}
		\left\Vert \frac{\sqrt{n}}{\sqrt{M_n}} \mathbb{P}_n(g_{x,t,D})- \frac{\mathbb{E}_{\fth_0}(g_{x,t,D})}{\sqrt{\alpha_{\fth_0}}} \right\Vert_{\mathcal{G}_D} + \left\Vert \frac{\mathbb{E}_{\hat{\fth}_n}(g_{x,t,D})}{\sqrt{\alpha_{\hat{\fth}_n}}} - \frac{\mathbb{E}_{\fth_0}(g_{x,t,D})}{\sqrt{\alpha_{\fth_0}}} \right\Vert_{\mathcal{G}_D}
	\end{equation}
	Adding $\pm 1/\sqrt{\alpha_{\fth_0}}$ in the first norm results in
	\begin{align*}
		&\left\Vert \left(\frac{\sqrt{n}}{\sqrt{M_n}}\pm \frac{1}{\sqrt{\alpha_{\fth_0}}}\right) \mathbb{P}_n(g_{x,t,D})- \frac{\mathbb{E}_{\fth_0}(g_{x,t,D})}{\sqrt{\alpha_{\fth_0}}} \right\Vert_{\mathcal{G}_D}\\
		&\leq \left\vert\frac{\sqrt{n}}{\sqrt{M_n}}- \frac{1}{\sqrt{\alpha_{\fth_0}}}\right\vert \cdot \left\Vert \mathbb{P}_n(g_{x,t,D}) \right\Vert_{\mathcal{G}_D}
		+ \frac{1}{\sqrt{\alpha_{\fth_0}}} \left\Vert \mathbb{P}_n(g_{x,t,D})- \mathbb{E}_{\fth_0}(g_{x,t,D}) \right\Vert_{\mathcal{G}_D}.
	\end{align*}
	From the LLN, we get
	\begin{equation*}
		\left\vert\frac{\sqrt{n}}{\sqrt{M_n}}- \frac{1}{\sqrt{\alpha_{\fth_0}}}\right\vert \overset{P} \longrightarrow 0 
	\end{equation*}
	and
	\begin{equation*}
		\left\Vert \mathbb{P}_n(g_{x,t,D}) \right\Vert_{\mathcal{G}_D}=\sup_{(x,t) \in D} \left\vert \frac{1}{n}\sum_{i=1}^n \mathds{1}_{\{(X_i,T_i)^\text{T} \in [0,x]\times[0,t]\cap D\}}\right\vert = \frac{M_n}{n} \overset{P} \longrightarrow  \alpha_{\fth_0}.
	\end{equation*}
	As demonstrated in Corollary \ref{cor2}, the class $\mathcal{G}_D$ exhibits the characteristics of a Glivenko-Cantelli and a Donsker class. Hence, it can be deduced that
	\begin{equation*}
		\left\Vert \mathbb{P}_n(g_{x,t,D})- \mathbb{E}_{\fth_0}(g_{x,t,D}) \right\Vert_{\mathcal{G}_D} \overset{a.s.} \longrightarrow 0.
	\end{equation*}
	The subsequent step involves the analysis of the second norm in \eqref{e6}. Assuming the validity of Assumption \ref{A2} and due to Corollary \ref{cor1}, the map $\fth \in \Theta \mapsto \mathbb{E}_{\fth} \in \ell^{\infty}(\mathcal{G}_D)$ is Fréchet differentiable at $\fth_0$ and consequently continuous in $\fth_0$, i.e.
	\begin{equation*}
		\Vert \mathbb{E}_{\fth}-\mathbb{E}_{\fth_0}\Vert_{\mathcal{G}_D} \to 0, \quad \text{as } \fth \to \fth_0. 
	\end{equation*}
	Furthermore, the map $\fth \mapsto 1/\sqrt{\alpha_{\fth}}$ is continuous in $\fth_0$ and interpreted as a map to $\ell^{\infty}(\mathcal{G}_D)$, it is constant. Hence, the composition $\fth \mapsto \mathbb{E}_{\fth}/\sqrt{\alpha_{\fth}}$ is continuous in $\fth_0$ also. The consistency $\hat{\fth}_n \to \fth_0$ in Assumption \ref{A3} and continuous mapping theorem \cite[][Theorem 18.11]{vaart1998}, together ensure that the second norm in \eqref{e6} converges to zero in probability as $n \to \infty$. \qed

\section{Hadamard differentiability}\label{App2}
Given that a function $\phi_{\Delta}$ is continuously differentiable in the common sense, the Hadamard differentiability for certain spaces is a direct consequence \cite[see][Problem 20.6]{vaart1998}.

\begin{lemma}\label{l4}\textbf{}\\
	Let $\phi_{\Delta}:[0,1] \to \mathbb{R}$ be a continuously differentiable function. Then, the map $z \mapsto \phi_{\Delta} \circ z$, whose domain is the set of functions $z:\mathcal{G} \to [0,1]$ contained in $\ell^{\infty}(\mathcal{G})$, is Hadamard differentiable.
\end{lemma}
\begin{lemma}\label{l3}\textbf{}\\
	Let $\phi^1_{\Delta}:[0,1] \to \mathbb{R}$ be a continuously differentiable function with derivative $(\phi^1_{\Delta})'$. Further denote $\phi_{\Delta}(y_1,y_2)=(y_1, \phi^1_{\Delta}(y_2))$ for $y_1, \,y_2 \in [0,1]$. Let the space $(\ell^{\infty}(\mathcal{G})^2, \Vert \cdot \Vert_{\textup{max}})$ be equipped with the maximum norm. Then, the mapping $z=(z_1,z_2) \mapsto \phi_{\Delta} \circ z$ for functions $z_1, \, z_2:\mathcal{G} \to [0,1]$ from $\ell^{\infty}(\mathcal{G})$ is Hadamard differentiable with 
	$$\phi'_{\Delta}(z_1,z_2)=\begin{pmatrix}
		1 & 0 \\
		0 & (\phi^1_{\Delta})'(z_2)
	\end{pmatrix}.$$
\end{lemma}
\begin{proof}
	Fix $z=(z_1,z_2) \in \ell^{\infty}(\mathcal{G})^2$. For $t \to 0$ and $h_t=(h_t^1,h_t^2) \to h=(h^1,h^2)$, according to Lemma \ref{l4}, it is
		\begin{align*}
			&\left\Vert \frac{\phi_{\Delta}(z+t h_t)-\phi_{\Delta}(z)}{t}-\phi_{\Delta}'(z)h\right\Vert_{\textup{max}}\\
			&\quad=\left\Vert
			\begin{pmatrix}
				(z_1+th_t^1-z_1)/t \\
				(\phi^1_{\Delta}(z_2+th_t^2)-\phi^1_{\Delta}(z_2))/t
			\end{pmatrix}
			-
			\begin{pmatrix}
				h^1\\
				(\phi^1_{\Delta})'(z_2)h^2
			\end{pmatrix}
			\right\Vert_{\textup{max}}\\
			&\quad=\max\left\{\vert h_t^1-h^1\vert,\left\vert (\phi^1_{\Delta}(z_2+th_t^2)-\phi^1_{\Delta}(z_2))/t-(\phi^1_{\Delta})'(z_2)h^2\right\vert\right\} \to 0.
	\end{align*}\qed
\end{proof}

\section{General form of covariance function}\label{App3}

The calculation of the covariance from equation \eqref{e9} is presented below. Define $g_i:=g_{x_i,t_i,D}=\mathds{1}_{[0,x_i]\times[0,t_i]\cap D} \in \mathcal{G}_D$ for $i=1,\,2$. In the first step, we insert the definition of the process $\mathbb{H}_{\mathbb{P}_{\fth_0}}$ from Theorem \ref{th1} into $\textup{Cov}[\mathbb{H}_{\mathbb{P}_{\fth_0}}g_1,\mathbb{H}_{\mathbb{P}_{\fth_0}}g_2]$, i.e.
\begin{align*}
	\mathbb{H}_{\mathbb{P}_{\fth_0}}g_i&=\mathbb{B}_{\mathbb{P}_{\fth_0}} (g_i)-\left(\mathbb{B}_{\mathbb{P}_{\fth_0}}(\phi_{\fth_0,1}),\mathbb{B}_{\mathbb{P}_{\fth_0}}(\phi_{\fth_0,2})\right) \cdot
	\begin{pmatrix}
		\dot{\mathbb{E}}_{\theta_0}(g_i)\\
		\dot{\mathbb{E}}_{\vth_0}(g_i)
	\end{pmatrix}\\
	&\quad +\left[\left(\mathbb{B}_{\mathbb{P}_{\fth_0}}( \phi_{\fth_0,1}), \mathbb{B}_{\mathbb{P}_{\fth_0}}( \phi_{\fth_0,2})\right)  \cdot
	\begin{pmatrix}
		\dot{\alpha}_{\theta_0}\\
		\dot{\alpha}_{\vth_0}
	\end{pmatrix}
	- \mathbb{B}_{\mathbb{P}_{\fth_0}}(\mathds{1}_D) \right]\cdot \frac{\mathbb{E}_{\fth_0}(g_i)}{\alpha_{\fth_0}}\\
	&=\mathbb{B}_{\mathbb{P}_{\fth_0}} (g_i)+ \left(\mathbb{B}_{\mathbb{P}_{\fth_0}}(\phi_{\fth_0,1}),\mathbb{B}_{\mathbb{P}_{\fth_0}}(\phi_{\fth_0,2})\right) \cdot \left[\frac{\mathbb{E}_{\fth_0}(g_i)}{\alpha_{\fth_0}}
	\begin{pmatrix}
		\dot{\alpha}_{\theta_0}\\
		\dot{\alpha}_{\vth_0}
	\end{pmatrix}
	-
	\begin{pmatrix}
		\dot{\mathbb{E}}_{\theta_0}(g_i)\\
		\dot{\mathbb{E}}_{\vth_0}(g_i)
	\end{pmatrix}\right]\\
	&\quad-\mathbb{B}_{\mathbb{P}_{\fth_0}}(\mathds{1}_D) \cdot \frac{\mathbb{E}_{\fth_0}(g_i)}{\alpha_{\fth_0}}.
\end{align*}
To enhance readability, define further $K_{i,j}:=\left(\mathbb{E}_{\fth_0}g_i\frac{\dot{\alpha}_{j}}{\alpha_{\fth_0}}-\dot{\mathbb{E}}_{j}g_i\right)$, for $i=1,2$ and $j \in \{\theta_0, \vth_0\}$.
Thus, the covariance is
\begin{align*}
	&\textup{Cov}\Bigg[\mathbb{B}_{\mathbb{P}_{\fth_0}} g_1+\mathbb{B}_{\mathbb{P}_{\fth_0}}\phi_{\fth_0,1}\cdot K_{1,\theta_0}+\mathbb{B}_{\mathbb{P}_{\fth_0}}\phi_{\fth_0,2} \cdot K_{1,\vth_0}-\mathbb{B}_{\mathbb{P}_{\fth_0}}\mathds{1}_D\cdot \frac{\mathbb{E}_{\fth_0}g_1}{\alpha_{\fth_0}},\\
	& \hspace{28mm}\mathbb{B}_{\mathbb{P}_{\fth_0}} g_2+\mathbb{B}_{\mathbb{P}_{\fth_0}}\phi_{\fth_0,1} \cdot K_{2, \theta_0}+\mathbb{B}_{\mathbb{P}_{\fth_0}}\phi_{\fth_0,2}\cdot K_{2,\vth_0}-\mathbb{B}_{\mathbb{P}_{\fth_0}}\mathds{1}_D \cdot\frac{\mathbb{E}_{\fth_0}g_2}{\alpha_{\fth_0}}\Bigg].
\end{align*}
The calculation of the covariance is performed sequentially by fixing one summand at a time and combining it with all opposing summands.\\
First summand:
\begin{align*}
	&\textup{Cov}\left[\mathbb{B}_{\mathbb{P}_{\fth_0}} g_1, \mathbb{B}_{\mathbb{P}_{\fth_0}} g_2\right]+K_{2,\theta_0} \cdot \textup{Cov}\left[\mathbb{B}_{\mathbb{P}_{\fth_0}} g_1, \mathbb{B}_{\mathbb{P}_{\fth_0}} \phi_{\fth_0,1} \right]\\
	&\quad+K_{2,\vth_0} \cdot \textup{Cov}\left[\mathbb{B}_{\mathbb{P}_{\fth_0}} g_1, \mathbb{B}_{\mathbb{P}_{\fth_0}} \phi_{\fth_0,2} \right]
	-\frac{\mathbb{E}_{\fth_0}g_2}{\alpha_{\fth_0}}\cdot \textup{Cov}\left[\mathbb{B}_{\mathbb{P}_{\fth_0}} g_1, \mathbb{B}_{\mathbb{P}_{\fth_0}}\mathds{1}_D \right]
\end{align*}
Second summand:
\begin{align*}
	&K_{1,\theta_0}\cdot \textup{Cov}\left[\mathbb{B}_{\mathbb{P}_{\fth_0}} g_2, \mathbb{B}_{\mathbb{P}_{\fth_0}} \phi_{\fth_0,1} \right]
	+K_{1,\theta_0}K_{2,\theta_0}\cdot\textup{Cov}\left[\mathbb{B}_{\mathbb{P}_{\fth_0}}  \phi_{\fth_0,1}, \mathbb{B}_{\mathbb{P}_{\fth_0}} \phi_{\fth_0,1} \right]\\
	&\quad +K_{1,\theta_0}K_{2,\vth_0}\cdot\textup{Cov}\left[\mathbb{B}_{\mathbb{P}_{\fth_0}}  \phi_{\fth_0,1}, \mathbb{B}_{\mathbb{P}_{\fth_0}} \phi_{\fth_0,2} \right]
	-\frac{\mathbb{E}_{\fth_0}g_2}{\alpha_{\fth_0}}K_{1,\theta_0}\cdot  \textup{Cov}\left[\mathbb{B}_{\mathbb{P}_{\fth_0}}  \phi_{\fth_0,1}, \mathbb{B}_{\mathbb{P}_{\fth_0}} \mathds{1}_D  \right]
\end{align*}
Third summand:
\begin{align*}
	&K_{1,\vth_0}\cdot \textup{Cov}\left[\mathbb{B}_{\mathbb{P}_{\fth_0}} g_2, \mathbb{B}_{\mathbb{P}_{\fth_0}} \phi_{\fth_0,2} \right]
	+K_{1,\vth_0}K_{2,\theta_0}\cdot\textup{Cov}\left[\mathbb{B}_{\mathbb{P}_{\fth_0}}  \phi_{\fth_0,1}, \mathbb{B}_{\mathbb{P}_{\fth_0}} \phi_{\fth_0,2} \right]\\
	&\quad +K_{1,\vth_0}K_{2,\vth_0}\cdot\textup{Cov}\left[\mathbb{B}_{\mathbb{P}_{\fth_0}}  \phi_{\fth_0,2}, \mathbb{B}_{\mathbb{P}_{\fth_0}} \phi_{\fth_0,2} \right]
	-\frac{\mathbb{E}_{\fth_0}g_2}{\alpha_{\fth_0}}K_{1,\vth_0}\cdot  \textup{Cov}\left[\mathbb{B}_{\mathbb{P}_{\fth_0}}  \phi_{\fth_0,2}, \mathbb{B}_{\mathbb{P}_{\fth_0}} \mathds{1}_D  \right]
\end{align*}
Fourth summand:
\begin{align*}
	&-\frac{\mathbb{E}_{\fth_0}g_1}{\alpha_{\fth_0}}\cdot \textup{Cov}\left[\mathbb{B}_{\mathbb{P}_{\fth_0}} g_2, \mathbb{B}_{\mathbb{P}_{\fth_0}}\mathds{1}_D \right]-\frac{\mathbb{E}_{\fth_0}g_1}{\alpha_{\fth_0}}K_{2,\theta_0}\cdot \textup{Cov}\left[\mathbb{B}_{\mathbb{P}_{\fth_0}} \phi_{\fth_0,1}, \mathbb{B}_{\mathbb{P}_{\fth_0}}\mathds{1}_D \right]\\
	&-\frac{\mathbb{E}_{\fth_0}g_1}{\alpha_{\fth_0}}K_{2,\vth_0}\cdot \textup{Cov}\left[\mathbb{B}_{\mathbb{P}_{\fth_0}} \phi_{\fth_0,2}, \mathbb{B}_{\mathbb{P}_{\fth_0}}\mathds{1}_D \right]
	+\frac{\mathbb{E}_{\fth_0}g_1\cdot \mathbb{E}_{\fth_0}g_2}{\alpha_{\fth_0}^2}\cdot \textup{Cov}\left[ \mathbb{B}_{\mathbb{P}_{\fth_0}}\mathds{1}_D, \mathbb{B}_{\mathbb{P}_{\fth_0}}\mathds{1}_D \right]
\end{align*}
The individual covariances contained therein can be simplified as follows, using the covariance of the $\mathbb{P}_{\fth_0}$-Brownian bridge from \eqref{e14} with $i,j =1,2$:
\begin{center}
	\setlength{\tabcolsep}{2pt}
	\begin{tabular}{lll}
		$\textup{Cov}\left[\mathbb{B}_{\mathbb{P}_{\fth_0}} g_1, \mathbb{B}_{\mathbb{P}_{\fth_0}} g_2\right]$
		&$=\mathbb{E}_{\fth_0}g_1g_2-\mathbb{E}_{\fth_0}g_1\mathbb{E}_{\fth_0}g_2$
		&\\
		$\textup{Cov}\left[\mathbb{B}_{\mathbb{P}_{\fth_0}} g_i, \mathbb{B}_{\mathbb{P}_{\fth_0}} \phi_{\fth_0,j}\right]$
		&$=\mathbb{E}_{\fth_0}g_i\phi_{\fth_0,j}-\mathbb{E}_{\fth_0}g_i\mathbb{E}_{\fth_0}\phi_{\fth_0,j}$
		&$=\mathbb{E}_{\fth_0}g_i\phi_{\fth_0,j}$\\
		$\textup{Cov}\left[\mathbb{B}_{\mathbb{P}_{\fth_0}} g_i, \mathbb{B}_{\mathbb{P}_{\fth_0}} \mathds{1}_D\right]$
		&$=\mathbb{E}_{\fth_0}g_i \mathds{1}_D-\mathbb{E}_{\fth_0}g_i\mathbb{E}_{\fth_0} \mathds{1}_D$
		&$=(1-\alpha_{\fth_0})\,\mathbb{E}_{\fth_0}g_i$\\
		$\textup{Cov}\left[\mathbb{B}_{\mathbb{P}_{\fth_0}}  \phi_{\fth_0,i}, \mathbb{B}_{\mathbb{P}_{\fth_0}} \phi_{\fth_0,j}\right]$
		&$=\mathbb{E}_{\fth_0}\phi_{\fth_0,i}\phi_{\fth_0,j}-\mathbb{E}_{\fth_0}\phi_{\fth_0,i}\mathbb{E}_{\fth_0}\phi_{\fth_0,j}$
		&$=\mathbb{E}_{\fth_0}\phi_{\fth_0,i}\phi_{\fth_0,j}$\\
		$\textup{Cov}\left[\mathbb{B}_{\mathbb{P}_{\fth_0}}  \phi_{\fth_0,i}, \mathbb{B}_{\mathbb{P}_{\fth_0}} \mathds{1}_D\right]$
		&$=\mathbb{E}_{\fth_0}\phi_{\fth_0,i}\mathds{1}_D-\mathbb{E}_{\fth_0}\phi_{\fth_0,i}\mathbb{E}_{\fth_0}\mathds{1}_D$
		&$=\mathbb{E}_{\fth_0}\phi_{\fth_0,i}\mathds{1}_D
		=0$\\
		$\textup{Cov}\left[\mathbb{B}_{\mathbb{P}_{\fth_0}}   \mathds{1}_D, \mathbb{B}_{\mathbb{P}_{\fth_0}} \mathds{1}_D\right]$
		&$=\mathbb{E}_{\fth_0}\mathds{1}_D^2-\mathbb{E}_{\fth_0}\mathds{1}_D\mathbb{E}_{\fth_0}\mathds{1}_D$
		&$=\alpha_{\fth_0}(1-\alpha_{\fth_0})$
	\end{tabular}
\end{center}
The simplifications include $\mathbb{E}_{\fth_0}\phi_{\fth_0,i}\mathds{1}_D
=\mathbb{E}_{\fth_0}\phi_{\fth_0,i}=0$ (see Assumption \ref{A4}) and $\mathbb{E}_{\fth_0}\mathds{1}_D=\alpha_{\fth_0}$. Inserting the values of the covariances leads to:\\
First and fourth summand:
\begin{align*}
	&\mathbb{E}_{\fth_0}g_1g_2-\mathbb{E}_{\fth_0}g_1\mathbb{E}_{\fth_0}g_2+K_{2,\theta_0} \cdot\mathbb{E}_{\fth_0}g_1\phi_{\fth_0,1}
	+K_{2,\vth_0} \cdot \mathbb{E}_{\fth_0}g_1\phi_{\fth_0,2}\\
	&-\frac{\mathbb{E}_{\fth_0}g_2}{\alpha_{\fth_0}}\cdot (1-\alpha_{\fth_0})\,\mathbb{E}_{\fth_0}g_1
	-\frac{\mathbb{E}_{\fth_0}g_1}{\alpha_{\fth_0}}\cdot (1-\alpha_{\fth_0})\,\mathbb{E}_{\fth_0}g_2
	+\frac{\mathbb{E}_{\fth_0}g_1\cdot \mathbb{E}_{\fth_0}g_2}{\alpha_{\fth_0}^2}\cdot \alpha_{\fth_0}(1-\alpha_{\fth_0})\\
	&=\mathbb{E}_{\fth_0}g_1g_2-\mathbb{E}_{\fth_0}g_1\mathbb{E}_{\fth_0}g_2/\alpha_{\fth_0}+K_{2,\theta_0} \cdot\mathbb{E}_{\fth_0}g_1\phi_{\fth_0,1}
	+K_{2,\vth_0} \cdot \mathbb{E}_{\fth_0}g_1\phi_{\fth_0,2}
\end{align*}
Second and third summand:
\begin{align*}
	&K_{1,\theta_0}\cdot \mathbb{E}_{\fth_0}g_2\phi_{\fth_0,1}
	+K_{1,\theta_0}K_{2,\theta_0}\cdot\mathbb{E}_{\fth_0}\phi_{\fth_0,1}\phi_{\fth_0,1}
	+K_{1,\theta_0}K_{2,\vth_0}\cdot\mathbb{E}_{\fth_0}\phi_{\fth_0,1}\phi_{\fth_0,2}\\
	&+K_{1,\vth_0}\cdot \mathbb{E}_{\fth_0}g_2\phi_{\fth_0,2}
	+K_{1,\vth_0}K_{2,\theta_0}\cdot\mathbb{E}_{\fth_0}\phi_{\fth_0,1}\phi_{\fth_0,2}
	+K_{1,\vth_0}K_{2,\vth_0}\cdot\mathbb{E}_{\fth_0}\phi_{\fth_0,2}\phi_{\fth_0,2}	
\end{align*}
Setting $K_i=(K_{i,\theta_0}, K_{i,\vth_0})^{\textup{T}}$, we finally obtain
\begin{align*}
	\textup{Cov}[\mathbb{H}_{\mathbb{P}_{\fth_0}}g_1,\mathbb{H}_{\mathbb{P}_{\fth_0}}g_2]&=
	\mathbb{E}_{\fth_0}g_1g_2-\mathbb{E}_{\fth_0}g_1\mathbb{E}_{\fth_0}g_2/\alpha_{\fth_0}+K_2^{\textup{T}} \mathbb{E}_{\fth_0}[g_1\phi_{\fth_0}]\\
	&\quad+K_1^{\textup{T}}\mathbb{E}_{\fth_0}[g_2 \phi_{\fth_0}]+K_1^{\textup{T}} \mathbb{E}_{\fth_0}[\phi_{\fth_0}\phi_{\fth_0}^{\textup{T}}]K_2.
\end{align*}

\section{Elements of covariance function for models in Sections \ref{sec61} and \ref{sec62}}\label{App4}

\subsection{For product copula ($\boldsymbol{\Pi}$)} \label{App41}

With $g_{x_i,t_i,D}=\mathds{1}_{[0,x_i]\times[0,t_i]\cap D}$, it is $\mathbb{E}_{\theta_0}(g_{x_i,t_i,D})=\mathbb{P}_{\theta_0}((X_1,T_1)^{\textup{T}} \in [0,x_i]\times[0,t_i]\cap D)$. We construct the intersection of the rectangle and the parallelogramm as a rectangle minus one or two triangles, depending on the position of $(x_i,t_i)$ relative to $D$. Of course, algebraic expressions are faster for the computation than numerical integration. With the following auxiliary functions, we calculate the expectation for every case, where $func_1$ is the integral of the density $f_{\theta_0}$ over the rectangle, and $func_2$ and $func_3$ are the integrals over the corresponding lower and upper rectangle, respectively. 
\begin{align*}
	func_1(x,t)&=\frac{t}{G}(1-e^{-\theta_0x})\\
	func_2(x)&=\frac {e^{-\theta_0 x}}{G \theta_0}(s\theta_0-\theta_0x -1)  +\frac {e^{-\theta_0 s}}{G \theta_0}\\
	func_3(t)&=\frac {e^{-\theta_0 t}}{G \theta_0}-\frac{1-t \theta_0}{G \theta_0}
\end{align*}	
\begin{align*}
	\mathbb{E}_{\theta_0}(g_{x_i,t_i,D})=
	\left\{\begin{array}{ll}
		func_1(x_i,t_i)-func_2(x_i)
		-func_3(t_i), &  (x_i,t_i) \in D, \, x_i>s\\
		func_1(x_i,t_i)-func_3(t_i), & (x_i,t_i) \in D, \, x_i \leq s\\
		func_1(t_i+s,t_i)-func_2(t_i+s)-func_3(t_i), & t_i < x_i-s\\
		func1(x_i,x_i)-func_2(x_i)-func_3(x_i), & x_i > s, \, x_i < t_i\\
		func_1(x_i,x_i)-func_3(x_i), &x_i \leq s, \, x_i < t_i
	\end{array}\right.
\end{align*}
The Fr\'{e}chet derivative has the representation according to Corollary \ref{cor1}, with $x$ and $t$ replaced by $x_i$ and $t_i$.
Therefore, we only have to derive our density with respect to $\theta$, while the geometric considerations concerning the integration regions stay the same.
\begin{align*}
	dfunc_1(x,t)&=\frac{xt}{G}e^{-\theta_0x}\\
	dfunc_2(x)&=\frac {e^{-\theta_0 x}}{G}\left(x^2-sx+\frac{x}{\theta_0} +{1}{\theta_0^2}\right)  -\frac {e^{-\theta_0 s}}{G \theta_0}\left(s+\frac{1}{\theta_0}\right)\\
	dfunc_3(t)&=\frac{1}{G \theta_0^2}\left(1-e^{\theta_0 t} (1+\theta_0 t)\right)
\end{align*}	
\begin{align*}
	\dot{\mathbb{E}}_{\theta_0}(g_{x_i,t_i,D})=
	\left\{\begin{array}{ll}
		dfunc_1(x_i,t_i)-dfunc_2(x_i)
		-dfunc_3(t_i), &  (x_i,t_i) \in D, \, x_i>s\\
		dfunc_1(x_i,t_i)-dfunc_3(t_i), & (x_i,t_i) \in D, \, x_i \leq s\\
		dfunc_1(t_i+s,t_i)-dfunc_2(t_i+s)-dfunc_3(t_i), & t_i < x_i-s\\
		dfunc_1(x_i,x_i)-dfunc_2(x_i)-dfunc_3(x_i), & x_i > s, \, x_i < t_i\\
		dfunc_1(x_i,x_i)-dfunc_3(x_i), &x_i \leq s, \, x_i < t_i
	\end{array}\right.
\end{align*}
According to the simple form of $\psi_{\theta_0}$ in Equation \ref{e15} under independent truncation, we have
$
	\mathbb{E}_{\theta_0}(\dot{\psi}_{\theta_0})=\alpha_{\theta_0} \left(\frac{2}{\theta_0^2} -\frac{s^2 e^{-\theta_0s}}{(1-e^{-\theta_0 s})^2} - \frac{G^2 e^{-G \theta_0}}{(1-e^{G \theta_0})^2}\right)$.

\subsection{For $\textup{FGM}$ copula} \label{App42}

Adding the Farlie-Gumbel-Morgernstern dependence, the auxiliary functions depend on two parameters, $\vth$ and $\theta$.
\begin{align*}
	func_1(x,t)&=\frac{t}{G}\left(1-e^{-\theta_0x}\right)\cdot \left[1 + \vth_0 e^{-\theta_0 x} \left(1-\frac{t}{G}\right)\right]\\
	func_2(x)&=\frac{s-x}{G}e^{-\theta_0 x}+\frac{1}{\theta_0 G} \left(e^{-\theta_0 s} - e^{-\theta_0 x}\right)\\
	&\quad -\frac{\vth_0}{G}(x-s) e^{-\theta_0 x}\cdot \left(1-\frac{x-s}{G}\right) \cdot \left(e^{- \theta_0 x}-1\right)\\
	&\quad -\frac{\vth_0}{G}e^{-\theta_0 x}\left(1-\frac{2 (x-s)}{G}\right)  \cdot \left(\frac{1}{2} e^{- \theta_0 x} - 1\right)\\
	&\quad +\frac{\vth_0}{G} e^{-\theta_0 s} \left(\frac{1}{2} e^{- \theta_0 s} - 1\right)
	+\frac{2\vth_0}{G^2 \theta_0^2} e^{-\theta_0 x}\left(\frac{1}{4} e^{- \theta_0 x} - 1\right)\\
	&\quad -\frac{2\vth_0}{G^2 \theta_0^2}e^{-\theta_0 s} \left(\frac{1}{4} e^{- \theta_0 s} - 1\right)\\
	func_3(t)&=\frac{t}{G} + \frac{1}{\theta_0 G} \left(e^{-\theta_0 t}-1\right) +\frac{\vth_0}{\theta_0 G} e^{-\theta_0 t}\cdot  \left(1-\frac{2t}{G} \right) \cdot \left(\frac{1}{2} e^{-\theta_0 t} -1 \right)\\
	&\quad+ \frac{1}{2} \frac{\vth_0}{\theta_0 G}-\frac{2 \vth_0}{G^2 \theta_0^2} e^{-\theta_0 t} \left(\frac{1}{4} e^{-\theta_0 t} -1 \right) - \frac{3 \vth_0}{2G^2 \theta_0^2}
\end{align*}
\begin{align*}
	\mathbb{E}_{\fth_0}(g_{x_i,t_i,D})=
	\left\{\begin{array}{ll}
		func_1(x_i,t_i)-func_2(x_i)
		-func_3(t_i), &  (x_i,t_i) \in D, \, x_i>s\\
		func_1(x_i,t_i)-func_3(t_i), & (x_i,t_i) \in D, \, x_i \leq s\\
		func_1(t_i+s,t_i)-func_2(t_i+s)-func_3(t_i), & t_i < x_i-s\\
		func1(x_i,x_i)-func_2(x_i)-func_3(x_i), & x_i > s, \, x_i < t_i\\
		func_1(x_i,x_i)-func_3(x_i), &x_i \leq s, \, x_i < t_i
	\end{array}\right.
\end{align*}
Therefore, we have to include the derivations for both parameters, that is $\dot{\mathbb{E}}_{\vth_0}(g_{x_i,t_i,D})$ and $\dot{\mathbb{E}}_{\theta_0}(g_{x_i,t_i,D})$. The correspondig derivations for the auxiliary functions have been conducted with $R$ and are omitted here. Due to the information matrix equality, we have $\mathbb{E}_{\fth}(\dot{\psi}_{\fth_0})^{-1}=-\mathbb{E}_{\fth_0}(\psi_{\fth_0}\psi_{\fth_0}^{\textup{T}})$ with $\psi_{\theta_0}$ from Equation \ref{e19}. The integrations therein have to be done numerically.

\end{document}